\documentclass[AMA,12pt,STIX1COL]{WileyNJD-v2}

\usepackage{amsfonts,amsmath,amssymb,bm,xcolor,dirtytalk,natbib}
\usepackage{systeme}
\usepackage{econometrics}
\usepackage{graphics}
\usepackage{dsfont}
\usepackage{graphicx}
\usepackage{siunitx}
\usepackage[default]{lato}
\usepackage[T1]{fontenc}
\usepackage{cancel}
\usepackage{float}
\usepackage{booktabs}
\usepackage{pdfpages}
\usepackage{hyperref}
\usepackage{graphicx}
\usepackage{amsmath}
\usepackage{graphicx}
\usepackage{enumerate}
\usepackage{amsfonts}
\usepackage{physics,commath,derivative, algorithm, algpseudocode} 
\usepackage{url} 
\usepackage{enumitem}
\usepackage{subcaption}
\usepackage[margin=10mm]{caption}
\usepackage{comment}

\newcommand{\rP}{\mbox{P}}
\newcommand{\Cov}{\mbox{Cov}}

\newcommand{\EE}{\mathbb{E}}
\newcommand{\vvar}{\mbox{var}}

\newcommand{\bX}{{X}}
\newcommand{\bY}{{Y}}
\newcommand{\bZ}{{Z}}
\newcommand{\bW}{{W}}
\newcommand{\btheta}{{\theta}}
\newcommand{\bbeta}{{\beta}}
\newcommand{\bnu}{{\nu}}
\newcommand{\bgamma}{{\gamma}}

\newcommand{\br}[1]{\left(#1\right)}

\articletype{Research Article}%
\begin{document}
\title{Two-stage testing for interactions in a high dimensional setting}

\author[1]{Marianne A Jonker}
\author[2]{Luc van Schijndel}
\author[2]{Eric Cator}

\titlemark{Two-Stage Testing in a high dimensional setting}
\authormark{Jonker \textsc{et al}}

\address[1]{\orgdiv{Research Institute for Medical Innovation, Science department IQ Health, Section Biostatistics}, \orgname{Radboud University Medical Center}, \orgaddress{\state{Nijmegen}, \country{Netherlands}}}

\address[2]{\orgdiv{Faculty of Science, Mathematics Department}, \orgname{Radboud University}, \orgaddress{\state{Nijmegen}, \country{Netherlands}}}

\corres{Marianne Jonker, Radboudumc Nijmegen, the Netherlands \email{marianne.jonker@radboudumc.nl}}



\abstract[Abstract]
{In a high dimensional regression setting in which the number of variables ($p$) is much larger than the sample size ($n$), the number of possible two-way interactions between the variables is immense. If the number of variables is in the order of one million, which is usually the case in  e.g., genetics, the number of two-way interactions is of the order one million squared. In the pursuit of detecting two-way interactions, testing all pairs for interactions one-by-one is computational unfeasible and the multiple testing correction will be severe. In this paper we describe a two-stage testing procedure consisting of a screening and an evaluation stage. It is proven that, under some assumptions, the tests-statistics in the two stages are asymptotically  independent. As a result, multiplicity correction in the second stage is only needed for the number of statistical tests that are actually performed in that stage. This increases the power of the testing procedure. Also, since the testing procedure in the first stage is computational simple, the computational burden is lowered. Simulations have been performed for multiple settings and regression models (generalized linear models and Cox PH model) to study the performance of the two-stage testing procedure. The results show type I error control and an increase in power compared to the procedure in which the pairs are tested one-by-one.  

\bigskip

\noindent%
}

\maketitle

\section{Introduction}
Let, for an individual, $Y$ be an outcome of interest and ${\bX}^t=(x_1,\ldots,x_p)$ a vector of covariates. Our aim is to detect covariates which are associated with the outcome $Y$ via an interaction. We have independent data of $n$ individuals from a population of interest, with $n\ll p$ (the number of individuals is much smaller than the number of covariates). This research question is important in, for instance, genetics and genomics, where the number of covariates are usually in the order of magnitude
of several millions. 

In genetics, the conventional approach to searching for pairwise interactions (between genetic markers) is to fit regression models and test for two-way interactions for all covariate pairs one-by-one. The number of pairs grows rapidly with the number of covariates. If data of one million covariates is available, the number of covariate pairs is of the order of one million squared. This is not only computational a burden, but the multiple testing correction will wipe out all effects that are present in the data. A wide variety of possible strategies have been proposed to search for interactions between genetic markers and between markers and environmental factors (see e.g., Kooperberg et al (2008), Murcray et al (2009), Gauderman et al (2010), Dai et al (2012),  Hsu et al (2012), Lewinger et al  (2013), Pecanka et al (2017),  Wong (Chapter 6, 2021), Johnsen et al (2021), Kawaguchi et al (2022), Kawaguchi et al (2023)). \nocite{Dai,Gauderman,Hsu,Johnsen,Kawaguchi,KawaguchiGE,Kooperberg,Lewinger,Murcray,Pecanka,PecankaB} An interesting approach that tackles both the computational problem and the multiplicity issue, is a two-stage approach. The first stage is seen as a screening stage to identify and select covariates that are more likely to be involved in an interaction. Only the selected covariates are passed on to the second stage. Subsequently, in the second stage (the verification stage), only pairs of covariates that were selected in the first stage are tested for interaction. 

By selecting a minor part of the covariates in the first stage, the computational burden of the whole testing procedure is reduced provided that the screening in the first stage is computationally simple. Even more important, with a two-stage procedure the multiplicity correction in the second stage may be reduced, which could result in a considerable power gain. If the statistical procedures in the two stages are statistically independent, the family wise error rate (FWER) of the whole testing procedure is controlled if multiplicity correction in the second stage is applied for the tests that are actually performed in the second stage. Independence between stages can be obtained in multiple ways. The easiest way is to split the data, one part for the first stage and the other part for the second stage. Another way to obtain independence is ``independence by design''; the test-statistics in both stages use all data, but are, nevertheless, statistically independent. A third option is to obtain independence by regressing out all the information in the data that was used in the first stage from the test-statistics in the second stage (see Pecanka et al (2017)\nocite{Pecanka}). An overview of multiple two-stage testing procedures with independent stages for a binary outcome is given by Wong, 2021 (Chapter 6, by Pecanka et al).\nocite{PecankaB} 

Kooperberg et al (2008) proposed a two-stage procedure for generalized linear models (GLMs). In the first stage, for every covariate a regression model with only the marginal effect is fitted and association between the covariate and the outcome variable is tested. If the corresponding p-value is below a pre-specified threshold, the covariate is passed on to the second stage. In the second stage interaction is tested for every pair of selected covariates. This is done by fitting a multivariable regression model and testing for an interaction effect. Dai et al (2012) adjusted the procedure in order to detect interactions between genetic markers and environmental factors.\nocite{Dai} In the first stage they only made a selection of the genetic markers and passed on the environmental factors to the second stage without testing. They proved FWER control of their procedure. In this paper, we focus on testing for interactions between genetic markers in GLMs and the Cox model. Although the proposed testing procedure is similar to the one in Dai et al (2012), proving control of the FWER is more complex, because in Dai et al (2012) the environmental factors are not screened in the first stage. We prove FWER control, by showing asymptotic independence between the testing procedures in the two stages. 
We also prove independence between the two stages (and thus FWER control) for the linear regression model with unknown variance of the error term (this proof of FWER control seems to be missing in Dai et al (2012)). For the semi-parametric Cox proportional hazards model for time-to-event outcomes, the proof of independence between stages relies on asymptotic properties of the partial likelihood.  

The remainder of the paper is organized as follows. In the sections \ref{sec: GLM} and \ref{sec: CoxPHmodel}  the theory of the two-stage testing procedure is explained for GLMs and the Cox PH model, respectively. To study control of the FWER and the power gain, multiple simulation studies have been performed for the linear,  Poisson and Cox models. This, and an application to bladder cancer data, are described in Section \ref{sec: simstudies}. In Section \ref{sec: discussion} assumptions and further research is discussed. The paper ends with multiple  appendices with proofs of theorems and lemmas.

\section{Two-stage testing procedure}
\label{sec: GLM}
\subsection{Two-stages}
To study the association between an outcome of interest $Y$ with an interaction between variables $x_k$ and $x_\ell$, the data are usually modeled with a GLM: 
\begin{align}
h\big(\mathbb{E}(Y|(x_k,x_\ell),\bbeta,\bnu)\big) = \beta_0 + \beta_1 x_{k} + \beta_2 x_\ell + \beta_3\; x_{k} x_\ell,
\label{BigModel}
\end{align}
where $h$ is the (chosen) link function and $\beta_0, \ldots, \beta_3$ are the unknown regression parameters of which $\beta_3$ reflects the  interaction effect. The unknown parameter $\bnu$ is an unknown nuisance parameters, e.g., the variance of the error term in the linear regression model.  The parameters can be estimated by their maximum likelihood estimators and the null hypothesis $H_0: \beta_3=0$ tested with, for instance, a Wald test. 

In a high dimensional setting with many variables $x_1,\ldots,x_p$ compared to the sample size $n$, $p \gg n$, it would not be sensible to test all possible pairs of variables one-by-one, as the number of pairs grows quadratically with the
number of variables. Kooperberg et al (2008) proposed the following two stage approach.\nocite{Kooperberg}  

\medskip

\noindent
{\bf{Stage 1, screening stage $S_1$}}\\
In the first stage all variables of ${\bX}^t=(x_1,\ldots,x_p)$ are screened for association with the outcome $Y$ in a marginal model. More specific, for variable $x_k$ the model
\begin{align}
h\big(\mathbb{E}(Y|{\bX},\bgamma=(\gamma_0,\gamma_1),\bnu)\big) &= \gamma_0 + \gamma_1 x_k
\label{SmallModel}
\end{align}
is fitted and the null hypothesis H$_0:\gamma_1=0$ is tested with the Wald test-statistic 
$T_k^{S_1}=\hat\gamma_1/(\widehat{\vvar}\; \hat\gamma_1)^{1/2}$ at a chosen significance level $\alpha_1$. Here $\hat\gamma_1$ is the maximum likelihood estimator of $\gamma_1$ and $\widehat{\vvar}\; \hat\gamma_1$ an estimator of its corresponding variance. 
If the null hypothesis is rejected, variable $x_k$ is passed on to the second stage, the verification stage. This is done for every variable in $\bX$. 

\medskip

\noindent
{\bf{Stage 2, verification stage $S_2$}}\\
Interaction for all pairs of variables in stage 2 are tested. More specific, suppose the variables $x_k$ and $x_\ell$ were passed on to the second stage, then the model
\begin{align}
h\big(\mathbb{E}(Y|{\bX},\bbeta=(\beta_0,\ldots,\beta_3),\bnu)\big) = \beta_0 + \beta_1 x_{k} + \beta_2 x_\ell + \beta_3 x_{k} x_\ell,
\label{FullModel}
\end{align}
is fitted and the null hypothesis H$_0:\beta_3=0$ is tested with the Wald test-statistics $T_{k\ell}^{S_2}=\hat\beta_3/(\widehat{\vvar}\; \hat\beta_3)^{1/2}$ at a chosen significance level $\alpha_2=\alpha/K_1$, with $\hat\beta_3$ the maximum likelihood estimator of $\beta_3$, 
$K_1$ the number of pairs in stage 2 and $\alpha$ the level of the test. This is done for every pair of variables in stage 2.\\

With this two stage procedure, the computational burden decreases as the test in the first stage is computational simpler than in the second stage (due to dimension reduction of the parameter space) and interaction is tested for a fraction of the variable pairs. Moreover, this approach is shown to be more powerful than a single stage testing approach in which all pairs are tested and a severe multiple testing correction is applied. This is discussed in the next subsection.

We will consider a more general version of the two stage procedure; covariates like gender, ethnicity or on behavior (e.g., smoking) that may be associated with the outcome variable, but are not of interest for the interaction study can be included in the regression models in both stages. Therefore, the models in the two stages will be reformulated later.

\subsection{Independent stages and multiple testing correction}
\label{subsec:independence}
To control the family wise error rate (FWER) of the two-stage testing procedure, a multiple testing correction is applied in the second stage. This multiple testing correction needs to take into account the fact that screening has taken place in the first stage; viewing the data and screening for interesting effects in the first stage cannot go unpunished. However, if under the combined null hypothesis all test-statistics in the second stage are statistically independent of those in the first stage, multiple testing correction in the second stage is only needed for the number of tests that are performed in that stage, as if no screening has taken place in stage 1. That means that if $K_1$ pairs are tested in the verification stage, the Bonferroni method will control the FWER of the whole procedure at a level $\alpha$ by using $\alpha_2=\alpha/K_1$ as a corrected significance level in the second stage. When the test-statistics in the two stages are not statistically independent under the combined null hypothesis, the control of the FWER is not guaranteed.\footnote{The number of pairs in the second stage, $K_1$, is stochastic and a function of the test-statistics in the first stage. This makes the two-stages dependent. We will come back to this Subsection \ref{Full independence GLM}.} Statistical independence  between the two stages under the combined null hypothesis is quite complex as ``pairwise independence'' is not sufficient to control the type I error of the two stage testing procedure; we need ``complete independence''. 

\medskip

\noindent
{\bf{Pairwise independence.}}
Suppose that for the variables $x_k$ and $x_\ell$ the  model in (\ref{SmallModel}) is fitted and the null hypotheses of no main effects are tested with the test-statistics $T^{S_1}_k$ and $T^{S_1}_\ell$, respectively. If both null hypotheses are rejected, it is tested whether the interaction effect, $\beta_3$ in the model (\ref{FullModel}), equals zero with the test-statistic $T^{S_2}_{k\ell}$. For pairwise independence the test-statistic $T_{k\ell}^{S_2}$ must be independent of both $T_k^{S_1}$ and $T_\ell^{S_1}$, for all $k$ and $\ell$. 

\medskip

\noindent
{\bf{Complete independence.}}
Let ${\cal K}$ be the set of indices of the variables that were passed on to the second stage. For complete independence, all test-statistics in the second stage, $T_{k\ell}^{S_2}, k,\ell\in {\cal K}$, should be independent of all test-statistics $T_{j}^{S_1}, j=1,\ldots,p$ in the first stage jointly. Complete independence implies pairwise independence, but the opposite is not true.
 
\medskip

If pairwise and complete independence hold, the value of the significance level in the first stage, $\alpha_1$, can be chosen freely. For high values of $\alpha_1$ many variables will be passed on to the second stage and the multiple testing correction in the second stage will be severe. The most extreme case is $\alpha_1=1$ when no selection takes place in the first stage. For a low value of $\alpha_1$ only a few variables will be passed on and the multiple testing correction in the second stage is minor. Although this may sound ideal, there exists the risk that variables that are associated with the outcome via an interaction do not pass the first stage, which will lower the power of the procedure. One may opt to use multiple thresholds in the first stage. For variables which are a priori of more interest a higher value for $\alpha_1$ (or even 1) can be taken than for variables that do not have this preferred position.

\subsection{Pairwise independence between stages}
\label{Pair-wise GLM}

Before focusing on pairwise independence for the two stage testing procedure, we formulate and prove pairwise independence in a more general setting. More specific, the regression models are formulated without the interaction term, but with non-specified extra covariates. These covariates can be chosen to be equal to an interaction term. Consider the three nested generalized linear regression models $M_1, M_2$ and $M_{true}$:
\begin{align}
\label{MM1}
M_1:\qquad h\big(\mathbb{E}(Y|\bX,\bbeta_1,\bnu_1)\big) &= \beta_{1,1} x_{1} + \beta_{1,2} x_{2} + \ldots + \beta_{1,p_1} x_{p_1}\\
\label{MM2}
M_2:\qquad h\big(\mathbb{E}(Y|\bX,\bbeta_2,\bnu_2)\big) &= \beta_{2,1} x_{1} + \beta_{2,2} x_{2} + \ldots + \beta_{2,p_1} x_{p_1} + \ldots + \beta_{2,p_2} x_{p_2}\\
\label{MM0}
M_{true}:\qquad h\big(\mathbb{E}(Y|\bX,\bbeta_0,\bnu_0)\big) &= \beta_{0,1} x_{1} + \beta_{0,2} x_{2} + \ldots + \beta_{0,p_1} x_{p_1} + \ldots + \beta_{0,p_2} x_{p_2} + \ldots + \beta_{0,p_0} x_{p_0}
\end{align}
with ${\bX}=(x_1,x_2,\ldots,x_p)^t$, $h$ the link function, $p_1\leq p_2 \leq p_0 \leq p$ and $x_1=1$ so that $\beta_{1,1}, \beta_{2,1}$ and $\beta_{0,1}$ are the intercepts in the models. The first index in $\beta_{i,j}$ refers to the model and the second index to the covariate. Model $M_1$ is nested in model $M_2$ and $M_2$ is nested in model $M_{true}$, with the latter model assumed to be the true model. This can always be effectuated, even if variables of the models $M_1$ and $M_2$ are not associated with the outcome, by including variables in the true model (if they are not in the model yet) with regression parameters equal to zero. Assume that the conditional density $f_{\btheta_0}$ of $Y|\bX$,  distributed according the true model $M_{true}$, belongs to a GLM with a canonical link function and, therefore, can be written as:
\begin{align}
f_{\btheta_0}(Y|\bX)= g(Y)\exp\big(V(Y)^t \bZ \btheta_0 - A(\bZ \btheta_0)\big), 
\label{expFamily}
\end{align}
where $\bZ$ is a function of the covariates $\bX$ and $\btheta_0$ the vector of unknown parameters in the true model; the regression and nuisance parameters. The vector $\btheta_0 \in \mathbb{R}^{q_0}$ with $q_0-p_0$ the number of nuisance parameters in the model. 

Let $\btheta_1$ and $\btheta_2$ be the parameter vectors for the models $M_1$ and $M_2$ (i.e., regression and nuisance parameters). Their parameter spaces are assumed to be subsets of $\mathbb{R}^{q_1}$ and $\mathbb{R}^{q_2}$, respectively, with $q_1\geq p_1$ and $q_2\geq p_2$ depending on the number of nuisance parameters. If no nuisance parameters are present $\theta_1, \theta_2$ and $\theta_0$ correspond to the vector of regression parameters. The parameter spaces for the models $M_1, M_2$ and $M_{true}$ are nested. To project a parameter vector from the space for $M_1$ or $M_2$ to the parameter space for model $M_{true}$ (the largest space), we define the projection matrices 
\begin{align*}
    E_1 & = \br{e_1|\ldots | \; e_{q_1}} \in \mathbb{R}^{q_0\times q_1}\\
    E_2 & = \br{e_1|\ldots | \; e_{q_2}} \in \mathbb{R}^{q_0\times q_2},
\end{align*}
in such a way that $\{e_1, \ldots, e_{q_1}\}$ and $\{e_1, \ldots, e_{q_2}\}$ form orthonormal bases for the parameter spaces for $M_1$ and $M_2$. To be more specific, these matrices represent linear functions from the parameter spaces for $M_i,i=1,2$ to $M_{true}$, which will be used to project the regression coefficients within the two submodels to the parameter space of the true model by left-multiplication. So, for arbitrary parameter vectors ${\theta}_1 \in \mathbb{R}^{q_1}$ and ${\theta}_2 \in \mathbb{R}^{q_2}$, it holds that $E_1{\theta}_1 \in \mathbb{R}^{q_0}$ and $E_2{\theta}_2 \in \mathbb{R}^{q_0}$. 

Let ${\bX}_1=(x_{1,1},\ldots,x_{1,p})^t, \ldots, {\bX}_n=(x_{n,1},\ldots,x_{n,p})^t$ be a sample of variables and let $Y_1|{\bX}_1, \ldots, Y_n|{\bX}_n$ be a sample from the GLM $M_{true}$. The maximum likelihood estimators for $\btheta_1$ and $\btheta_2$ in the models $M_1$ and $M_2$ are denoted as $\hat\btheta_1$ and $\hat\btheta_2$. Furthermore, let $\btheta_1^0$ and $\btheta_2^0$ be the values that maximize the pointwise limit of the log likelihood function corresponding to the different models (divided by $n$). The projections of the vectors $\hat\btheta_2$ and $\btheta_2^0$ in the parameter space of model $M_2$ to the space that is orthogonal to the parameter space of model $M_1$ are computed as $\hat\btheta_2^\perp=(Id_{q_2} - E_2^tE_1E_1^tE_2)\hat\btheta_2$ and $\btheta_2^{0,\perp}=(Id_{q_2} - E_2^tE_1E_1^tE_2)\btheta_2^0$, for $Id_{q}$ the $q\times q$ identity matrix. In the theorem below we will prove that $\sqrt{n}(\hat\btheta_1-\btheta_1^0)$ and $\sqrt{n}(\hat\btheta_2^\perp-\btheta_2^{0,\perp})$ are asymptotically independent if $M_2=M_{true}$. The proof of the theorem is given in Appendix A.\\

\noindent
\begin{theorem}
\label{theorem:pairwuseGLM}
Suppose that $Y_1|{\bX}_1, \ldots, Y_n|{\bX}_n$ is a sample from the generalized linear model with a canonical link function $M_{true}$ in (\ref{MM0}). Let $\hat{\btheta}_1$ and $\hat{\btheta}_2$ denote the maximum likelihood estimators of ${\theta}_1^0$ and ${\btheta}_2^0$ in the nested models $M_1$ and $M_2$, respectively. If model $M_2$ coincides with the true model $M_{true}$ (i.e., $p_2 = p_0$), then $\sqrt{n}(\hat\btheta_1-\btheta_1^0)$ and  $\sqrt{n}(\hat\btheta_2^\perp-\btheta_2^{0,\perp})$ are asymptotically independent.\\
\end{theorem}

Now, we translate the theorem to the two stage testing framework by reformulating the GLMs so that it includes an interaction term and explicitly defining the matrices $E_1$ and $E_2$. Consider the two nested GLMs $M_1$ and $M_2$:
\begin{align}
\label{M1}
M_1:\qquad h\big(\mathbb{E}(Y|{\bX},\bbeta_1,\bnu_1)\big) &= \beta_{1,1} x_{1} + \beta_{1,2} x_{2} + \beta_{1,k} x_{k}\\
\label{M2}
M_2:\qquad h\big(\mathbb{E}(Y|{\bX},\bbeta_2,\bnu_2)\big) &= \beta_{2,1} x_{1} + \beta_{2,2} x_{2} + \beta_{2,k} x_{k} + \beta_{2,\ell} x_{\ell} +  \beta_{2,k\ell} x_{k} x_{\ell}
\end{align}
with ${\bX}=(x_1,x_2,\ldots,x_p)^t$, $h$ the link function, $x_1=1$ so that $\beta_{1,1}$ and $\beta_{2,1}$ are the intercepts in the models, and $x_2$ the ``fixed'' covariate which is included in all models and is not tested (multiple ``fixed'' covariates can be included). The model $M_1$ equals the GLM that is fitted in the first stage, the screening stage $S_1$, of the two-stage testing procedure as defined earlier. In this stage it is studied whether the covariate $x_k$ (genetic marker $k$) is associated with the outcome $Y$ by testing the null hypothesis $H_0: \beta_{1,k}=0$. Model $M_2$ is the interaction model that is fitted in stage 2, the verification stage $S_2$. Now, it is tested whether the interaction term $x_k x_\ell$ is associated with $Y$, by testing the null hypothesis $H_0: \beta_{2,k\ell}=0$.

Further, $\bbeta_1=(\beta_{1,1},\beta_{1,2},\beta_{1,k})^t$ and $\bbeta_2=(\beta_{2,1},\beta_{2,2},\beta_{2,k},\beta_{2,\ell},\beta_{2,k\ell})^t$ denote the vectors of regression parameters in the two models, and their maximum likelihood estimators are denoted as $\hat{\bbeta}_1$ and $\hat{\bbeta}_2$. The vectors $\bbeta_1^0=(\beta_{1,1}^0,\beta_{1,2}^0,\beta_{1,k}^0)^t$ and $\bbeta_2^0=(\beta_{2,1}^0,\beta_{2,2}^0,\beta_{2,k}^0,\beta_{2,\ell}^0,\beta_{2,k\ell}^0)^t$ are the values that maximize the pointwise limit of the log likelihood function corresponding to the models $M_1$ and $M_2$ (divided by $n$). These vectors coincide with the true parameters only if the corresponding regression model equals the true model. Define $\hat\bbeta_2^\perp=(\hat\beta_{2,\ell},  \hat\beta_{2,k\ell})^t$ and $\bbeta_2^{0,\perp}=(\beta_{2,\ell}^{0}, \beta_{2,k\ell}^0)^t$, the regression parameters for the covariates that are additional to model $M_2$ compared to model $M_1$.  

If, in the model, there are no nuisance parameters  (e.g., the Poisson and logistic regression model), then $\btheta_1=\bbeta_1$ and $\btheta_2=\bbeta_2$. If, moreover, $\{e_1, \ldots, e_{q_1}\}$ and $\{e_1, \ldots, e_{q_2}\}$ form canonical bases, $\btheta_2^{0,\perp}=(0.\ldots,0,\theta_{q_1+1}^0,\ldots,\theta_{q_2}^0)^t$ and it follows from Theorem \ref{theorem:pairwise} that  $\sqrt{n}(\hat\btheta_1-\btheta_1^0)=\sqrt{n}(\hat\bbeta_1-\bbeta_1^0)$ and  $\sqrt{n}(\hat\btheta_2^\perp-\btheta_2^{0,\perp})=\sqrt{n}(\hat\bbeta_2^\perp-\bbeta_2^{0,\perp})$ are asymptotically independent. Thus, $\sqrt{n}(\hat\beta_{1,k}-\beta^0_{1,k})$ and $\sqrt{n}(\hat\beta_{2,k\ell}-\beta_{2,k\ell}^0)$ are asymptotically independent: the interaction effect $\hat \beta_{2,k\ell}$ in model $M_2$ (fitted in the second stage of the two-stage testing procedure) is asymptotically independent of the marginal effects $\hat \beta_{1,k}$ in model $M_1$ that is fitted in the first stage.  \\ 

\begin{result}[{\bf Pairwise independence of estimated regression coefficients}]
\label{theorem:pairwise} 
Suppose that $Y_1|{\bX}_1, \ldots, Y_n|{\bX}_n$ is a sample from the GLM $M_2$ in (\ref{M2}) with a canonical link function $h$. If there are no nuisance parameter, $\sqrt{n}(\hat\beta_1-\beta_1^0)$ and  $\sqrt{n}(\hat\bbeta_2^\perp-\bbeta_2^{0,\perp})$ are asymptotically independent.\\
\end{result}

Suppose that in the GLMs with a canonical link function, the parameter vector $\btheta_0$ as defined in the density of $Y|X$ in (\ref{expFamily}) is of the form $(\nu_0^t,\bbeta_0^t)^t$, with $\nu_0$ a nuisance parameter. By choosing $\{e_1, \ldots, e_{q_1}\}$ and $\{e_1, \ldots, e_{q_2}\}$ so that they form canonical bases for the parameter spaces for $M_1$ and $M_2$, Theorem \ref{theorem:pairwise} can be directly applied to prove pairwise independence; i.e., asymptotically independence between the interaction effect $\hat \beta_{2,k\ell}$ in the second stage and the marginal effects $\hat \beta_{1,k}$ and $\hat \beta_{1,\ell}$ in the first stage of the two-stage testing procedure. If there are nuisance parameters in the model, but the parameter vector $\btheta_0$ is not of the form $(\nu_0^t,\bbeta_0^t)^t$ (the regression and the nuisance parameters are not separated), like in the linear regression model (see below), pairwise independence of the estimated regression coefficients is not guaranteed by the theorem. As an example, suppose $Y|\bX$ has a normal distribution with mean $\bX^t \bbeta_0$ and variance $\sigma_0^2$. By writing the conditional density in the form (\ref{expFamily}): 
\begin{align*}
f_{\bbeta_0,\sigma_0^2}(y|\bX) \;=\; 
\frac{1}{\sqrt{2\pi}} \exp\Big(y \; \bX^t \bbeta_0/\sigma_0^{2}  - y^2 /(2\sigma_0^{2}) -(\bX^t \bbeta_0)^2/(2\sigma_0^2) - \log \sigma_0   \Big)  
\;=\; \frac{1}{\sqrt{2\pi}} \exp\Big( V(y)^t Z\theta_0-A(Z\theta_0) \Big)
\end{align*}
it follow that the model belongs to the family of GLMs with a canonical link function, with $V(y)=(y^2,y)^t$ and $A(Z\btheta_0)=(\bX^t\bbeta_0)^2/(2\sigma_0^2)+\log\sigma_0$, the parameter ${\btheta_0}= (-1/(2\sigma_0^2), {\bbeta_0}^t/\sigma_0^2)^t\in \mathbb{R}^{p+1}$, and 
\begin{align}
Z = \begin{pmatrix}
~1 &~~ {\bf 0}^t~\\
~0 &~~ {\bX}^t~
\end{pmatrix}
\in \mathbb{R}^{2}\times \mathbb{R}^{p+1} 
\label{Z}
\end{align}
for ${\bf 0}\in \mathbb{R}^{p}$ a vector with zeroes. By Theorem \ref{theorem:pairwuseGLM}, and for the standard basis, $\sqrt{n}(\hat\beta_{1,k}/\hat\sigma_1^2-\beta_{1,k}^0/\sigma_1^2)$ and $\sqrt{n}(\hat\beta_{2,k\ell}/\hat\sigma_2^2-\beta_{2,k\ell}^0/\sigma_2^2)$ are 
asymptotically independent if model $M_2$ is the true model. Here $\hat\sigma_1^2$ and $\hat\sigma_2^2$ are the estimated variances of the errors in the models $M_1$ and $M_2$, $\sigma_1^2$ and $\sigma_2^2$. From this independence results we can not (directly) conclude asymptotic independence of $\sqrt{n}(\hat\beta_{1,k}-\beta_{1,k}^0)$ and $\sqrt{n}(\hat\beta_{2,k\ell}-\beta_{2,k\ell}^0)$.  For the linear regression model pairwise independence can be proven by straightforward calculations. \\

\begin{theorem}[{\bf Pairwise independence in the linear regression model}]
\label{theorem:pairwiseGaussian} 
Suppose that $Y_1|{\bX}_1, \ldots, Y_n|{\bX}_n$ is a sample from the linear regression model $M_2$ in (\ref{M2}) with the identical link function and unknown variance $\sigma^2$. Then $\hat\beta_1$ and  $\hat\bbeta_2^\perp$ are independent.\\
\end{theorem}

The proof of Theorem \ref{theorem:pairwiseGaussian} is given in Appendix A. In terms of the two-stage testing procedure, the previous results state that if $M_2$ is the true model, the estimate of the regression parameter for the interaction term in the second stage, $\hat\beta_{2,k\ell}$ in the model $M_2$ in (\ref{M2}), is (asymptotically) independent of the maximum likelihood estimators of the regression parameter $\hat\beta_{1,k}$ in the model in (\ref{M1}) and by symmetry also of $\hat\beta_{1,\ell}$.   
  
The assumption that model $M_2$ is the true model is essential for the proof. This assumption implies that all variables that are associated with the outcome via a main effect or an interaction effect must be included in this model. This assumption would be difficult to hold in a study in which one aims to find new variables (interactions of genetic markers) that are associated with the outcome. However, to control the FWER in the two-stage testing procedure, independence is only required under the (full) null hypothesis. Therefore, in the two-stage procedure just described, the null hypothesis should state that there are no main and no interaction effects among the variables that participate in the two-stage testing procedure to be associated with the outcome; known main effects should be included as ``fixed'' covariates. The paper by Dai et al (2012) may make you think differently about the choice of the null hypothesis.\nocite{Dai} However, in their paper pairwise independence between the test-statistics in the two stages is proven under the (implicit) assumption of no model misspecification of the interaction model in (\ref{M2}), and thus under the assumption that all true main and interaction effects are included in this model for every statistical test in the verification stage. In other words, they (implicitly) work with the same null hypothesis. 
The above discussion also means that, once we can prove complete independence, we have FWER control in weak sense (and not in strong sense like is claimed by Dai et al (2012)\nocite{Dai}); under the union of all null hypotheses. A proof of Theorem \ref{theorem:pairwiseGaussian} for unknown variance $\sigma^2$ seems to be missing in Dai et al (2012).\nocite{Dai}


The theorems above state that there is pairwise independence between the maximum likelihood estimators of the regression parameters. However, in the two-stage testing procedure the test-statistic $T_k^{S_1}=\hat\beta_{1,k}/\widehat{\vvar} (\hat\beta_{1,k})^{1/2}$ for main effects in stage 1 has to be (asymptotically) independent of the test-statistic $T_{k\ell}^{S_2}=\hat\beta_{2,k\ell}/\widehat{\vvar} (\hat\beta_{2,k\ell})^{1/2}$ for testing for interaction in stage 2, under the null hypothesis that all main and interaction effects equal zero (so also $\beta_{1,k}^0=\beta_{2,k}^0=\beta_{2,\ell}^0=\beta_{2,k\ell}^0=0$. 
The Wald test-statistics in the first and second stage equal $\hat\beta_{1,k}/\widehat{\vvar} (\hat\beta_{1,k})^{1/2} = \sqrt{n}\hat\beta_{1,k} (I(\hat{\btheta}_1)^{-1}_{(k)})^{1/2}$ and $\hat\beta_{2,k\ell}/\widehat{\vvar} (\hat\beta_{2,k\ell})^{1/2} = \sqrt{n}\hat\beta_{2,k\ell} (I(\hat{\btheta}_2)^{-1}_{(k\ell)})^{1/2}$, with $I(\hat{\btheta}_1)$ and $I(\hat{\btheta}_2)$ the estimated Fisher information matrices in the models $M_1$ and $M_2$, and $I(\hat{\btheta}_2)_{(k\ell)}^{-1}$ equal to the diagonal element of the inverse Fisher Information matrix that corresponds to the parameter $\beta_{k\ell}$ (and similar for $I(\hat{\btheta}_1)_{(k)}^{-1}$). Asymptotic pairwise independence between the test-statistics follows from the pairwise independence of the maximum likelihood estimators of the regression parameters and Slutky's lemma (see e.g., van der Vaart (1998)).\nocite{vdV1998} \\

\begin{result}[{\bf Pairwise independence of test-statistics}]
\label{res: pairwise_teststatistics}
Suppose that $Y_1|{\bX}_1, \ldots, Y_n|{\bX}_n$ is a sample from model $M_2$ in (\ref{M2}) with a canonical link function $h$ without any nuisance parameter $\nu_2$, or from the linear regression model (identical link function) with unknown variance $\sigma^2$.
Then, under the full null hypothesis of no main and interaction effects, it holds that the Wald test-statistics $T_k^{S_1}$ and $T_{k\ell}^{S_2}$ are asymptotically independent.
\end{result}

\subsection{Complete independence between stages}
\label{Full independence GLM}

Define the generalized linear regression models:
\begin{align}
M_1:\qquad h\big(\mathbb{E}(Y|\bX,\bbeta_1,\bnu_1)\big) &= \beta_{1,1} x_{1} + \beta_{1,2} x_{2} + \beta_{1,j} x_{j} \label{FullM1}\\
M_2:\qquad h\big(\mathbb{E}(Y|\bX,\bbeta_2,\bnu_2)\big) &= \beta_{2,1} x_{1} + \beta_{2,2} x_{2} + \hspace{1.3cm} \beta_{2,k} x_{k} +  \beta_{2,\ell} x_{\ell} + \beta_{2,k\ell} x_{k} x_{\ell}\label{FullM2}\\
M_{true}:\qquad h\big(\mathbb{E}(Y|\bX,\bbeta_0,\bnu_0)\big) &= \beta_{0,1} x_{1} + \beta_{0,2} x_{2} \label{FullM3}
\end{align}
where, this time, model $M_1$ (fitted in stage 1) is not nested in model $M_2$ (fitted in stage 2), and the model $M_{true}$, the true model under the full null hypothesis, is nested in both $M_1$ and $M_2$. Every model has an intercept ($x_1=1$) and the ``fixed'' covariate $x_2$ (multiple ``fixed'' covariates can be included). 

Complete independence of the test-statistics can be proven under the null hypothesis of no main and no interaction effects for independent (standardized) covariates for the logistic, the Poisson and the linear regression model. This will be proven in steps. The first step is mixed pairwise independence. The proof of the theorem below is given in Appendix B. \\ 

\begin{theorem}[{\bf Mixed pairwise independence of estimated regression coefficients}]
\label{theorem:full}
Suppose that $Y_1|{\bX}_1, \ldots, Y_n|{\bX}_n$ is a sample from model $M_{true}$ in (\ref{FullM3}) and let $\hat{\bbeta}_1$ and $\hat{\bbeta}_2$ be the maximum likelihood estimators of the parameters ${\beta}_1=(\beta_{1,1}, \beta_{1,2}, \beta_{1,j})^t$ and ${\beta}_2=(\beta_{2,1}, \beta_{2,2}, \beta_{2,k},\beta_{2,\ell},\beta_{2,k\ell})^t$ in the models $M_1$ and $M_2$ in (\ref{FullM1}) and (\ref{FullM2}), respectively. Under the null hypothesis of no main and no interaction effects (Model $M_{true}$) and if the covariates $x_j, x_k, x_\ell$ are mutually independent, are independent of the covariate $x_2$ and have mean zero, then $\sqrt{n}(\hat\bbeta_1-\bbeta_1^0)$ and  $\sqrt{n}\hat\bbeta_2^\perp$ are asymptotically independent in the logistic, the Poisson and the linear regression model, with $\hat\bbeta_2^\perp=(\hat\beta_{2,k},\hat\beta_{2,\ell},\hat\beta_{2,k\ell})^t$, for all $j, k, \ell$ with $j\neq k,\ell$ and $\hat\bbeta_2^\perp=(\hat\beta_{2,\ell},\hat\beta_{2,k\ell})^t$ if $j=k$ (and similarly in case $k=\ell$).  \\
\end{theorem}

In the theorem, $\bbeta_1^0=(\beta_{1,1}^0,\beta_{1,2}^0,0)$ with $\beta_{1,1}^0=\beta_{0,1}^0$ and $\beta_{1,2}^0=\beta_{0,2}^0$, the true parameter values in model $M_{true}$. From the theorem it follows that $\sqrt{n}\hat \beta_{2,k\ell}$ and $\sqrt{n}\hat\beta_{1,j}$ are asymptotically independent. If $j=k$ or $j=\ell$ the theorem implies pairwise independence.

It directly follows from Theorem \ref{theorem:full} and Slutky's Lemma that the Wald test-statistics in the second stage are asymptotically independent of the Wald test-statistics in the first stage, under the null hypothesis of no main and interaction effects (and under the assumptions stated in theorem above). This leads to the following result. \\

\begin{result}[{\bf Mixed pairwise independence of test-statistics}]
Suppose that $Y_1|{\bX}_1, \ldots, Y_n|{\bX}_n$ is a sample from model $M_{true}$ in (\ref{FullM3}). Under the null hypothesis of no main and no interaction effects (Model $M_{true}$) and if the covariates $x_j, x_k, x_\ell$ are mutually independent, independent of the covariate $x_2$ and have mean zero, then the Wald test-statistics  $T_j^{S_1}$ and $T_{k\ell}^{S_2}$ are asymptotically independent in the logistic, the Poisson and the linear regression model, for all $j,k,\ell$.\\  
\end{result}

Complete independence is stronger than mixed pairwise independence of the test-statistics. For complete independence every test-statistic $T_{k\ell}^{S_2}$ in the second stage should be asymptotically independent of all first stage test-statistics $T_j^{S_1}, j=1,\ldots,p$ jointly. Since, under the null hypothesis, the test-statistics in both stages are asymptotically Gaussian, complete independence follows from the asymptotic (mixed) pairwise independence.  

Now, complete independence has been proven for the linear, logistic and the Poisson regression models. For other GLMs complete independence may be true as well, but this depends on the form of the model and needs to be checked for each model separately. More details on how to check complete independence can be found in Appendix B. 

The number of pairs in the verification stage, $K_1$, determines the multiple testing correction in stage 2. This number is stochastic and a function of the test-statistics in the screening-stage. In order to obtain independence between the two-stages, the number $K_1$ could be replaced by the non-random value $(p\alpha_1)^2$, which approximately equals the expected number of pairs in the second stage under the (combined) null hypotheses, so under a sparse setting. 

Complete independence holds under the assumption that the covariates are independent. In a GWAS the covariates (i.e., genetic markers) can not be assumed to be independent; markers are expected to be correlated with neighboring markers. This is only a fraction of the total number of markers in the data set. In the simulation studies in Subsection \ref{sub:simprocedure} it will be seen that in all considered settings with correlated markers the FWER is controlled.

Pairwise and complete independence imply FWER control of the proposed testing procedure. However, under exactly the same assumptions as made in Theorem \ref{theorem:full}, FWER control can also be proved directly and without replacing the stochastic number $K_1$ in the multiple testing correction. This proof is given in Appendix C. Summarized, the two-stage procedure in which a Bonferroni correction is performed in the second stage only, preserves FWER control at level $\alpha$ under the assumptions stated above.

\section{Two-stage testing with the Cox proportional hazards model}
\label{sec: CoxPHmodel}
The most popular model for time-to-event data is the Cox proportional hazards (PH) regression model. This model is a semi-parametric model and does not belong to the family of GLMs. The regression parameters in the model are estimated by maximizing the partial likelihood. Asymptotic pairwise and complete independence can be proved along the same line as before, but the proof is based on the partial likelihood function (in stead of the full likelihood).  

Let $\tilde T$ be the time to an event of interest, from a well defined time zero. Given a covariate vector $\bX$, the hazard function is defined as \\[-9mm] 
\begin{align*}
h(t|\bX) = h_0(t) \exp(\bbeta^t \bX),
\end{align*}
where $h_0$ is the completely unknown baseline hazard function and $\bbeta$ the vector of unknown regression coefficients. Let $C$ be the right-censoring time, which is assumed to be independent of time $\tilde T$. Furthermore, let $T=\min\{\tilde T,C\}$ and  $\Delta=1_{\{\tilde T \leq C\}}$ the indicator function which indicates whether the event has been observed or is right-censored. Suppose data of $n$ independent individuals is available: $\{(T_1,\Delta_1, \bX_1),\ldots,(T_n,\Delta_n,\bX_n)\}$, where the index refers to the individual. In the Cox regression model the vector with regression parameters $\bbeta$ is estimated by maximizing the partial likelihood:
\begin{align}
PL(\bbeta) = \frac{1}{n} \sum_{i=1}^n \int_0^\infty \log \Bigg(\frac{\exp(\bbeta^t \bX_i)}{\sum_{j=1}^n Y_j(t)\exp(\bbeta^t \bX_j)}\Bigg){\rm d} N_i(t),
\end{align}
with $Y_i(t)=1_{\{T_i>t\}}$ indicating whether a person is still at risk at time $t$ and $N_i(t)=1_{\{T_i\leq t\}}$ (see e.g., Aalen et al (2008) and Therneau et al (2000)\nocite{Aalen,Therneau}). 

Let $M_{true}$ be the true model under the full null hypothesis, and let $M_1$ and $M_2$ be the nested models with hazard functions
\begin{align}
\label{SM1}
M_1:\qquad h_1(t|\bX) &= h_{0,1}(t) \exp(\beta_{1,1} x_{1} + \beta_{1,k} x_{k}) \\
\label{SM2}
M_2:\qquad h_2(t|\bX) &= h_{0,2}(t) \exp(\beta_{2,1} x_{1} + \beta_{2,k} x_{k} + \beta_{2,\ell} x_{\ell} + \beta_{2,k\ell} x_{k} x_{\ell}), \\
M_{true}:\qquad h_0(t|\bX) &= h_{0,0}(t) \exp(\beta_{0,1} x_{1}
).\label{SM0}
\end{align}
Here, $x_1$ is the ``fixed" covariate as explained before. The intercept is missing, as it is absorbed in the baseline hazard function. The maximum likelihood estimators of the regression parameters in the models $M_1$ and $M_2$ are denoted as $\hat\bbeta_1=(\hat\beta_{1,1},\hat\beta_{1,k})^t$ and $\hat\bbeta_2=(\hat\beta_{2,1},\hat\beta_{2,k},\hat\beta_{2,\ell},\hat\beta_{2,k\ell})^t$, respectively. Further, let $\hat\bbeta_2^\perp=(\hat\beta_{2,\ell},\hat\beta_{2,k\ell})^t$ be the estimated  regression parameters of the covariates in model $M_2$, that are additional to the covariates in model $M_1$. Finally, let  $\bbeta_1^0=(\beta_{1,1}^0,0)^t$, with $\beta_{1,1}^0=\beta_{0,1}^0$ the parameter value in the true model $M_{true}$.\\

\begin{theorem}[{\bf Pairwise independence in the Cox PH model}] 
\label{Theorem:Cox pairwise} 
Let $(T_1=\tilde{T}_1\wedge C_1, \Delta_1, \bX_1), \ldots, (T_n=\tilde{T}_n\wedge C_1, \Delta_n, \bX_n)$ be a random sample. Suppose that $\tilde{T}_i|X_i$ is a random observation from the Cox PH model with hazard function (\ref{SM0}), then $\sqrt{n}(\hat\bbeta_1-\bbeta_1^0)$ and  $\sqrt{n}\hat\bbeta_2^\perp$ are asymptotically independent.\\
\end{theorem}

The proof is given in Appendix D. Complete independence for the Cox PH model can be shown as before, under the full null hypothesis and under the assumptions of independent (and standardized) covariates. To control the FWER in the two-stage testing procedures, the test-statistics in the two stages need to be asymptotically completely independent, under the null hypothesis. Like for the GLMs, this independence follows from the asymptotic (mixed) pairwise independence of the maximum likelihood estimators of the regression parameters and Slutky's Lemma (see e.g., Van der Vaart (1998)\nocite{vdV1998}).

\section{Simulation studies and application}
\label{sec: simstudies}
Suppose the data $(Y_1, \bX_1), \ldots, (Y_n, \bX_n)$ of $n$ independent individuals from a population of interest are available. For individual $i$ the variable $Y_i$ is the outcome of interest and $\bX_i=(x_{i,1},\ldots,x_{i,p})^t$ is a vector with a numerical representation of the genotypes at $p$ biallelic markers. The genotypes can be numerically represented in different ways depending on the genetic model that is assumed (dominant, additive, recessive), but for simplicity we assume that $x_{i,j}$ equals the number of minor alleles that individual $i$ has at marker $j$; so taking the values 0, 1, 2. Our aim is to detect markers for which the outcome of interest is associated with an interaction of markers. 

In this section we present the results of simulation studies for the linear, the Poisson  and the Cox PH model. Type I error control of the whole two-stage procedure is considered in Subsection \ref{sub:FWER}, and in Subsection \ref{sub:power} power comparisons for different settings (correlated and uncorrelated markers), different first stage thresholds and for different values of the interaction effects are made. In Subsection \ref{sub:marginal} we justify testing for marginal effects in the first stage to detect interaction effects. We start with a description of the simulation process in Subsection \ref{sub:simprocedure}.

\subsection{Simulation procedure}
\label{sub:simprocedure}
Data are simulated or obtained in two steps:
\begin{enumerate}
\item obtaining marker data by simulation or using real life data,
\item simulation of phenotype (outcome) data based on the marker data.
\end{enumerate} 
A single data set consists of marker and phenotype data of 2000 individuals with each 3000 markers. Observations of different individuals are assumed to be independent. 

\bigskip

\noindent
{\bf Marker data}\\
The FWER and the power of the two stage testing procedure are studied in two settings: the marker genotypes are either independent or correlated within individuals. 

\medskip

\noindent
{\it Independent marker data.} For the 2000 individuals with each 3000 markers the number of minor alleles are simulated independently from a binomial distribution with parameters $m=2$ (two alleles) and a probability that is drawn from a uniform distribution on the interval $[0.1 ; 0.5]$. For the power simulations, two markers are considered to be causal. Their number of minor alleles are drawn from the binomial distribution with parameters $m=2$ and probability $0.2$. The minor allele frequencies of the causal markers are fixed in order to diminish the variability in the data. 

\medskip

\noindent
{\it Correlated marker data.} The correlated marker data are real life genotype data to mimic a realistic correlation structure between the markers (Pecanka et al (2019)\nocite{Pecanka2}). The data set contains the data of 2000 individuals with each 3000 markers. The linkage disequilibrium (LD) structure of the 3000 markers is shown in Figure \ref{fig:LDstructure}.

\begin{figure}[h]
\begin{center}
\includegraphics[angle=90,width=0.035\textwidth]{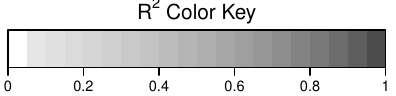} \includegraphics[height=0.25\textwidth, width=0.82\textwidth]{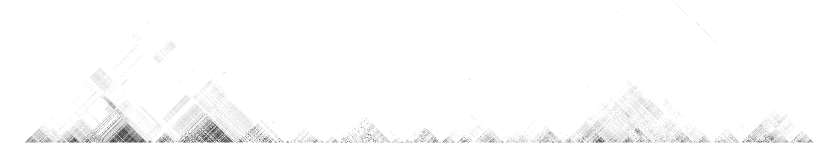} 
\caption{Heat map of the LD pattern for the 3000 markers based on data of 2000 individuals. LD is measured by the squared correlation coefficient.}
\label{fig:LDstructure}
\end{center}
\end{figure}

A few quality checks have been performed to be sure that significant results will not be the consequence of extremely low variability in the covariates or the interaction term, or of collinearities in the second stage. In the data sets with independent data hardly any marker failed to pass the quality check. 

\bigskip

\noindent
{\bf Phenotype data}\\
Let $(x_1, x_2)$ be the pair of causal markers. For given values of main and interaction effects $\beta_1, \beta_2$ and $\beta_3$, the phenotype (outcome) data are simulated from the models:
\begin{itemize}
\item Linear model. The outcome value is simulated from the model $Y=\beta_1 x_1 + \beta_2 x_2 + \beta_3 x_1 x_2 +\epsilon$ by simulating $\epsilon$ from the ${\mathcal N}(0,1)$ distribution. 
\item Poisson model with offset. The offset $t$ is simulated from a uniform[1, 5] distribution. Next, the outcome variable is simulated from a Poisson distribution with mean $\mu=t\;\exp(\beta_1 x_1 + \beta_2 x_2 +\beta_3 x_1 x_2)$.
\item Cox PH model. The hazards function is assumed to be equal to $\lambda(t)=\lambda_0(t)\exp(\beta_1 x_1 + \beta_2 x_2 +\beta_3 x_1 x_2)$ with the baseline hazard $\lambda_0$ so that the event times  follow a Weibull distribution with shape and scale parameters equal to 2 and 1, respectively, if $x_1 = x_2 =0$. The censoring times are simulated from a uniform distribution at the interval $[z_{\{0.70\}} ; z_{\{0.99\}}]$ with $z_{\{x\}}$ the $x$th quantile of the Weibull distribution with shape and scale parameters equal to 2 and 1. These numbers were chosen so that approximately 10\% of the observations are censored under the full null hypothesis that $\beta_1=\beta_2=\beta_3=0$.
\end{itemize}
Different values of the regression parameters $\beta_1, \beta_2$ and $\beta_3$ have been considered. In the simulation studies to study FWER control, $\beta_1=\beta_2=\beta_3=0$.

\subsection{Simulation results: FWER}
\label{sub:FWER}
For each model and multiple first stage thresholds, the FWER is studied for independent and dependent marker genotypes. This is done in the following steps:  1) if the markers are assumed to be independent, marker data are simulated as described before (otherwise the correlated marker data are used). 2) phenotype (outcome) data are simulated for 2000 individuals under the full null hypothesis of no main and no interaction effects ($\beta_1=\beta_2=\beta_3=0$). 3) the two-stage testing procedure is applied. A type I error is made if at least one significant interaction is found in the second stage. An interaction between two markers is significant in the second stage if its p-value $< 0.05/K_1$ with $K_1$ the number of pairs in the second stage. 4) the steps 2 and 3 are repeated 5000 times. The FWER is estimated by the fraction of data sets for which at least one significant interaction was found in the second stage. 
This procedure is performed for every model and several thresholds. 

The results for the three models, for uncorrelated and correlated markers and for the first stage thresholds equal to 0.05, 0.01 and 0.005 are presented in Table \ref{tab:type I error}. In case of uncorrelated markers, the FWERs are just below the threshold of 0.05 for the three models and different first stage thresholds. For the correlated markers, the FWERs are rather low; possibly the multiple testing correction is too conservative. In Kooperberg et al (2008)\nocite{Kooperberg2008} simulation studies had been performed for the logistic regression model. When studying the type I error, they simulated the phenotype data from a model that includes a main effect, but no interaction effect. The marker genotypes were simulated from a first order Markov chain with a  correlation between neighboring markers equal to 0.7. They also found that the Bonferoni correction is conservative. In a first order Markov chain, the correlation between the markers decreases fast with their distance. That might be the reason that our results are more conservative than their outcomes.

\begin{table}[h]
\centering
\begin{tabular}{|l|c|c|c|c|c|c|}
\hline
& \multicolumn{3}{c|}{uncorrelated markers} & \multicolumn{3}{c|}{correlated markers} \\ 
FST & Linear & Poisson & CoxPH & Linear & Poisson & CoxPH  \\ 
\hline  
0.05 & 0.0448 & 0.0416 & 0.0432 & 0.0076 & 0.0048 & 0.0064 \\ 
  0.01 & 0.0462 & 0.0492 & 0.0432 & 0.0114 & 0.0094 & 0.0120 \\ 
  0.005 & 0.0452 & 0.0364 & 0.0490 & 0.0078 & 0.0078 & 0.0084 \\
  \hline
\end{tabular}
\caption{Type I error for different first stage thresholds (FST) for uncorrelated and correlated markers.}
\label{tab:type I error}
\end{table}


\subsection{Simulation results: Power}
\label{sub:power}
For this simulation study two causal markers are chosen. In the "uncorrelated marker setting" the number of minor alleles at these causal markers are sampled from a binomial distribution with parameters $m=2$ and probability 0.2. The number of minor alleles for the remaining markers are sampled as was described before. 
For the correlated marker setting, two of the 3000 markers are chosen such that they have a correlation of 0.3 and minor allele frequencies equal to 0.48 and 0.25. These causal markers are used to simulate the phenotype data from the three models of interest as described in Subsection \ref{sub:simprocedure} with $x_1$ and $x_2$ the number of minor alleles at the two causal markers. The interaction effect $\beta_3$ runs from 0 to 1 in small steps of 0.05. Three settings for the main-effects are considered
\begin{enumerate}
    \item no main effects: $\beta_1=\beta_2=0$.
    \item two main effects: $\beta_1=\beta_2=0.5$.
    \item opposite main effects: $\beta_1=\beta_2=-0.5$.
\end{enumerate}
The remaining simulation assumptions are as described in Subsections \ref{sub:simprocedure}. For each model, each value of the interaction effect ($\beta_3$) and for each first stage threshold (FST $\in \{1, 0.1, 0.05, 0.01, 0.005\}$) new data are simulated and the two stage testing procedure is performed. The power is estimated as the fraction of times the causal pair of markers is significant in the second stage after Bonferroni correction. 

The results of the simulation studies for the uncorrelated markers in the linear and the Cox model are presented in Figure \ref{fig:power_uncorrelatedmarkers}. The results of the Poisson regression model can be found in Appendix E. If there are no main effects (only interaction effects, first column) or main effects in the same direction as the interaction effect (second column) the two stage approach is clearly more powerful than the approach in which all marker-marker combinations are tested for interaction (i.e., FST=1). In general (except for one setting) the power increases with a decreasing first stage threshold (FST). The lower the threshold in the first stage, the less markers are passed on to the second stage, and the less interaction terms are tested for association. This leads to a less severe multiple testing correction. However, the lower the threshold in the first stage, the lower the probability the causal markers are passed on to the second stage. The power plots for opposite main effects show a drop in power for large values of the interaction effect. This drop is not seen in the classical approach (FST=1) and therefore, for large effects, the strategy of testing all interaction pairs is more powerful than the two-stage testing procedure. It seems that large interaction effects cancel out the marginal association of the variables with the outcome, and prevent the markers to be selected in the first stage.

\begin{figure}[h]
\begin{center}
\includegraphics[height=0.32\textwidth,width=0.32\textwidth]{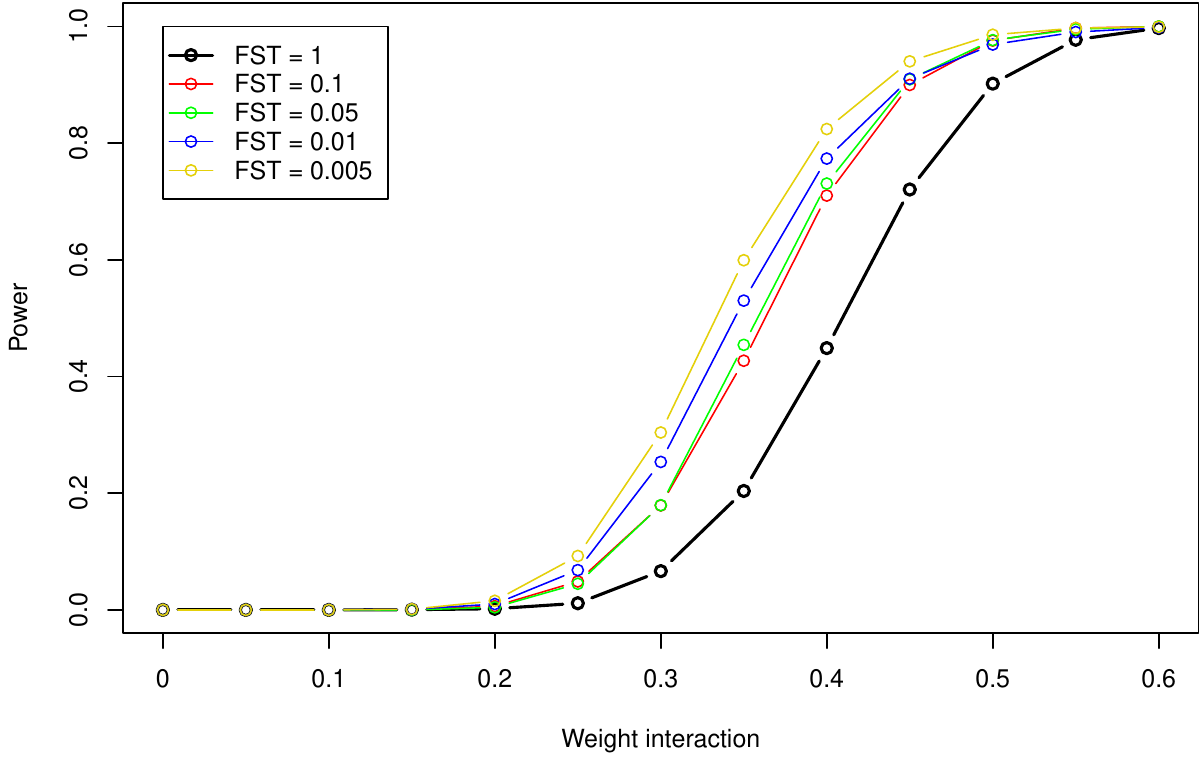} \;
\includegraphics[height=0.32\textwidth,width=0.32\textwidth]{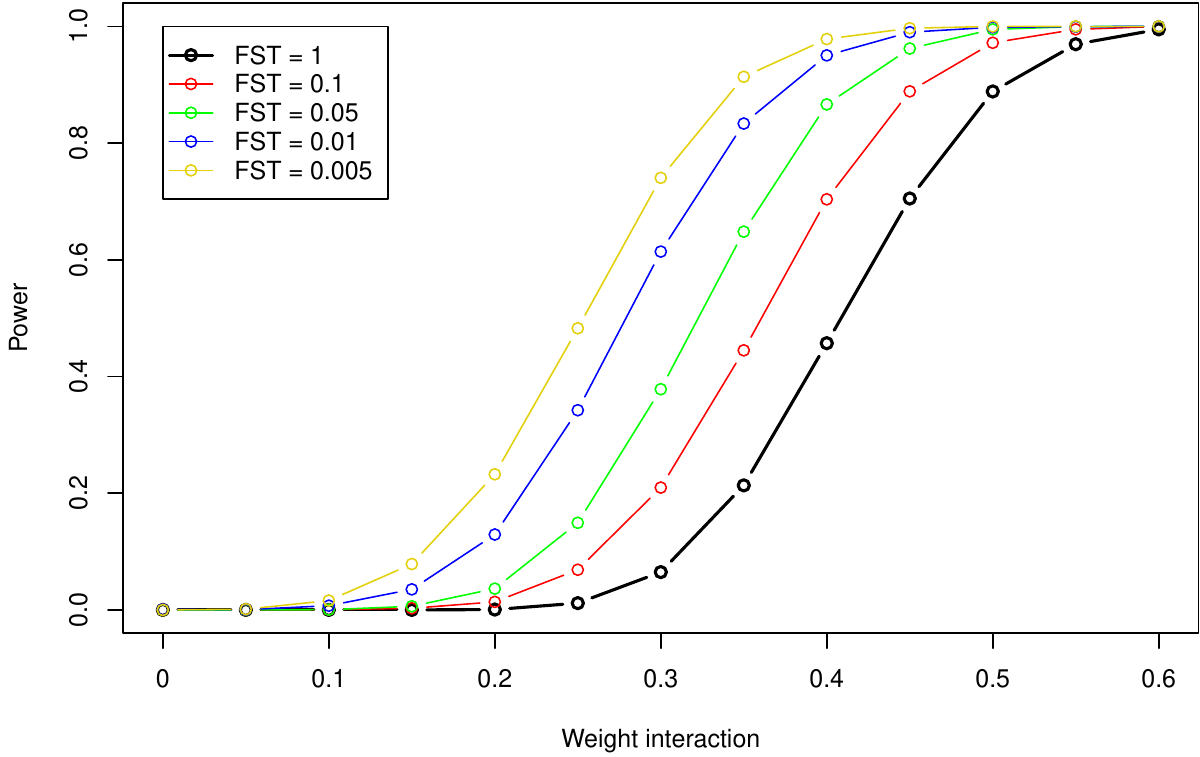} 
\;
\includegraphics[height=0.32\textwidth,width=0.32\textwidth]{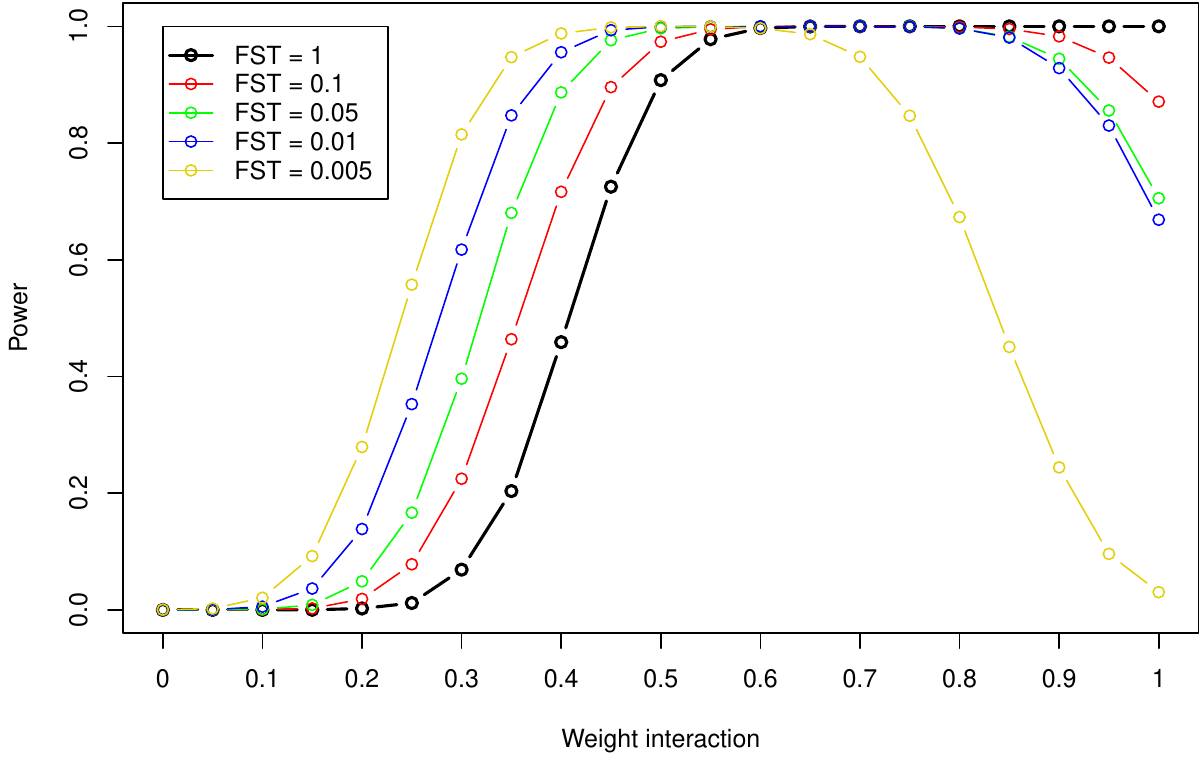} \\
\vspace{5mm}
\includegraphics[height=0.32\textwidth,width=0.32\textwidth]{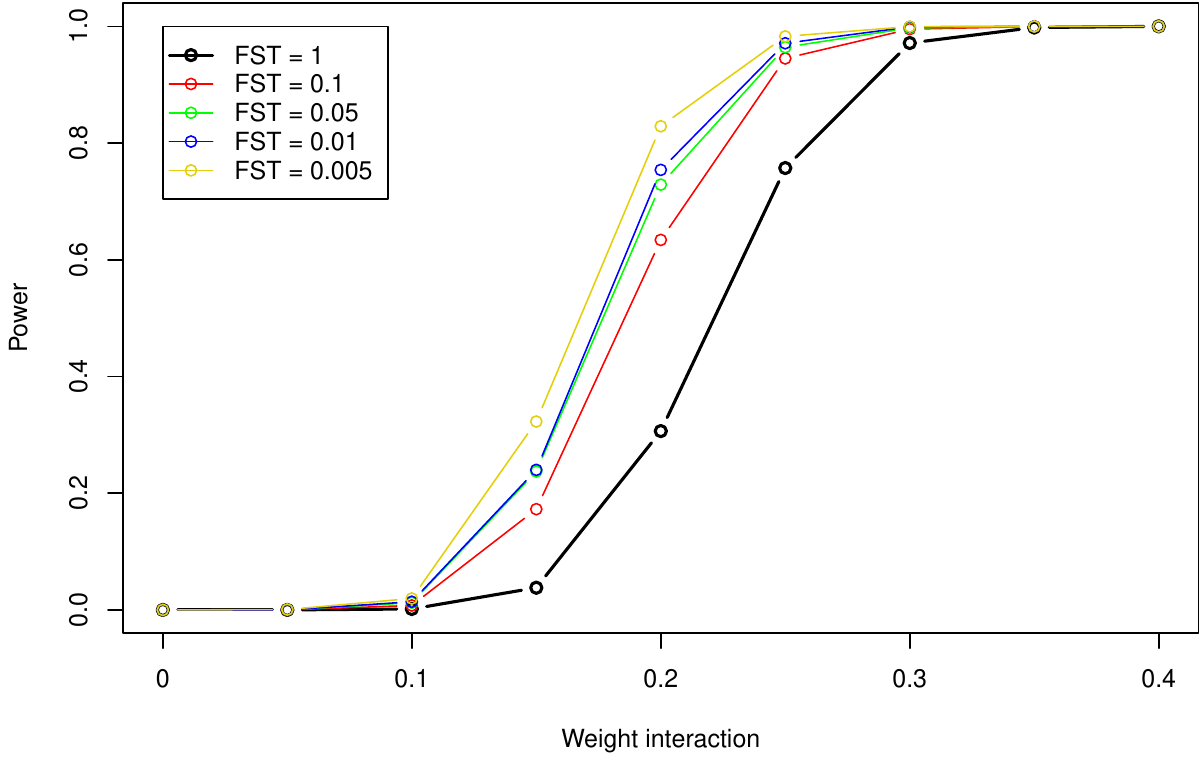}
\;
\includegraphics[height=0.32\textwidth,width=0.32\textwidth]{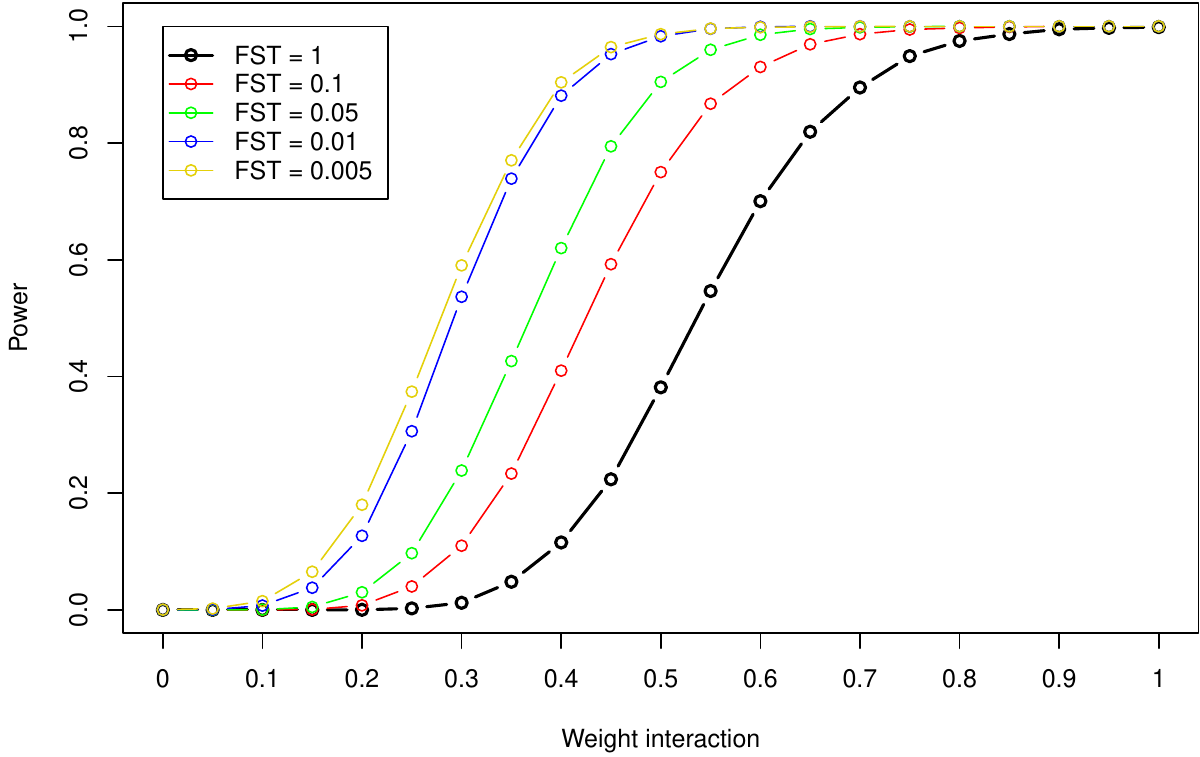}
\;
\includegraphics[height=0.32\textwidth,width=0.32\textwidth]{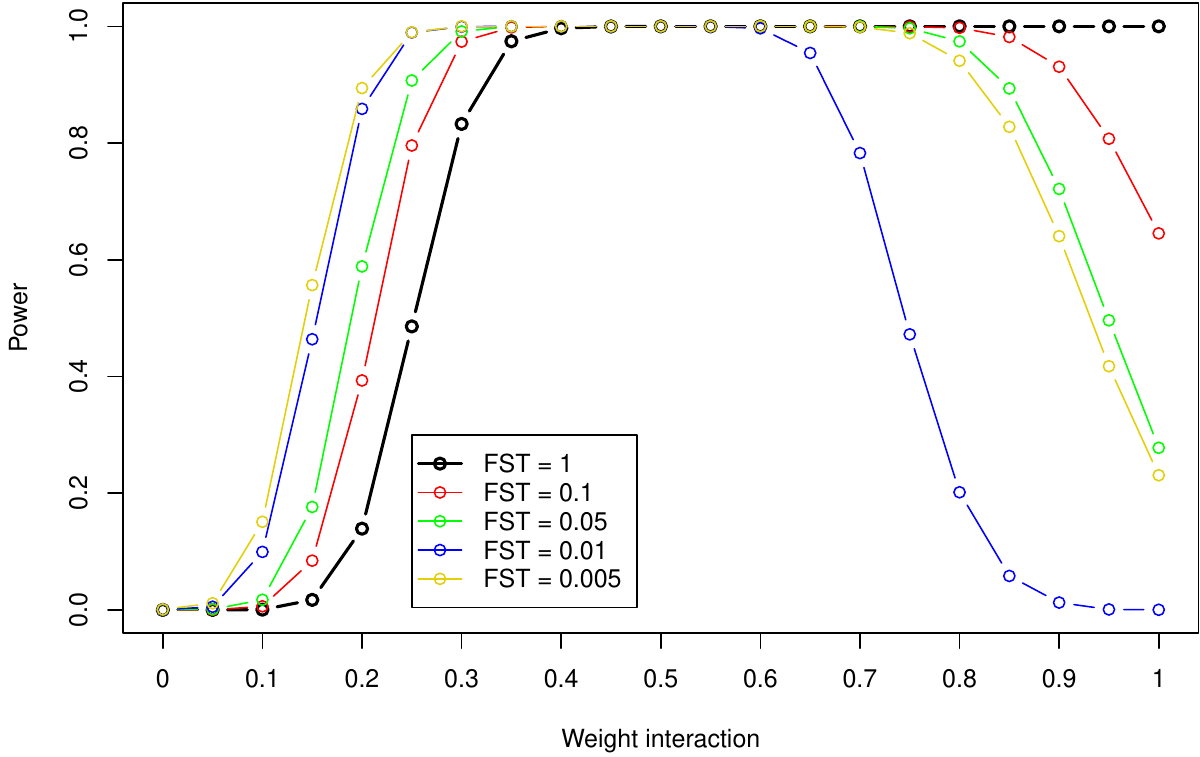}
\caption{Power as a function of the interaction effect for different first stage threshold (FST), with independent markers. First row: linear regression model; second row: Cox PH model. First column: no main effects; Second column: both main effects equal 0.5; Third column: both main effects equal -0.5. (Note the different scales on the x-axes.)   }
\label{fig:power_uncorrelatedmarkers}
\end{center}
\end{figure}

For the correlated marker data, the results for the linear regression model and the Cox model are shown in Figure \ref{fig:power_correlatedmarkers} and for the Poisson regression model they can be found in Appendix E.  The power gain of the two-stage testing procedure is smaller in the correlated setting compared to the setting with uncorrelated markers, but is still clearly visible, except in the setting in which the marginal effects are opposite to the interaction effect. For first stage thresholds smaller than 1, the power curves corresponding to the different first stage threshold are almost equal.   

A quality check has been performed in the second stage. Markers that are strongly mutually correlated or with their interaction term ($r^2 \geq 0.9$) are not tested for an interaction effect and are left out from the analysis. The check was performed to prevent unreliable results due to collinearity. The two causal markers have a low correlation ($r^2=0.3$).     

\begin{figure}[h]
\begin{center}
\includegraphics[height=0.32\textwidth,width=0.32\textwidth]{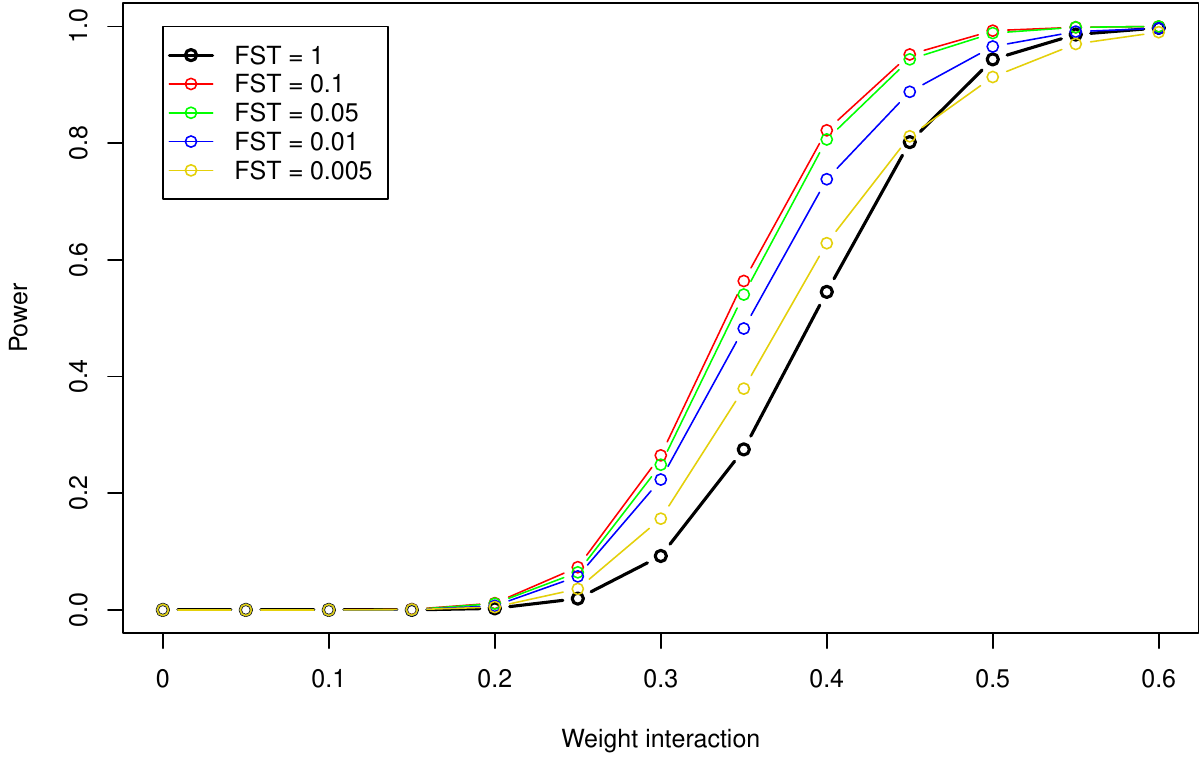} 
\;
\includegraphics[height=0.32\textwidth,width=0.32\textwidth]{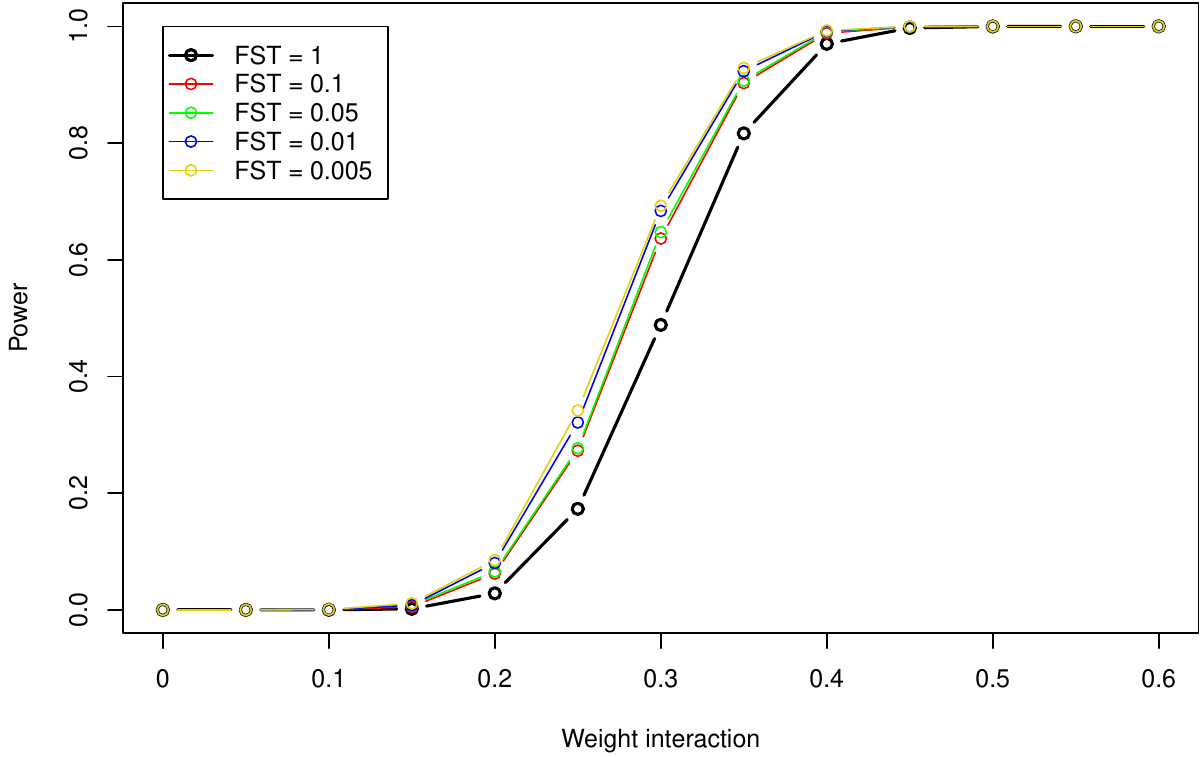} 
\;
\includegraphics[height=0.32\textwidth,width=0.32\textwidth]{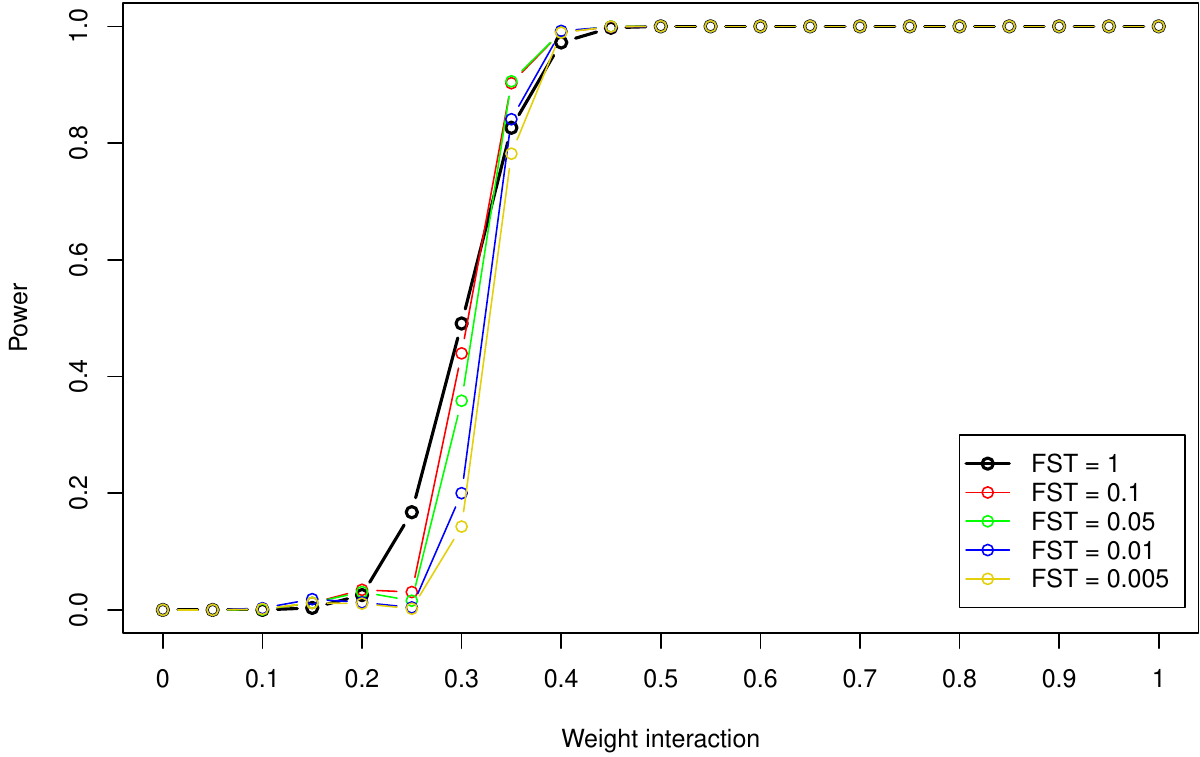} \\
\vspace{5mm}
\includegraphics[height=0.32\textwidth,width=0.32\textwidth]{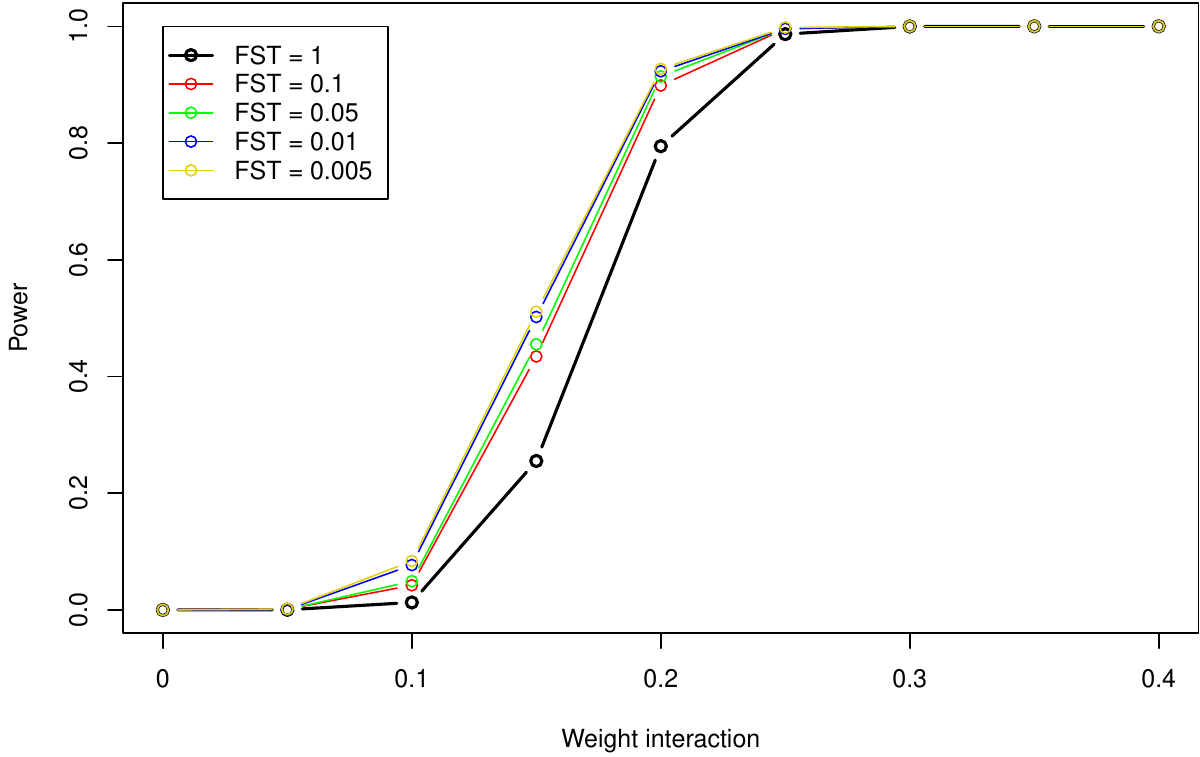}\;
\includegraphics[height=0.32\textwidth,width=0.32\textwidth]{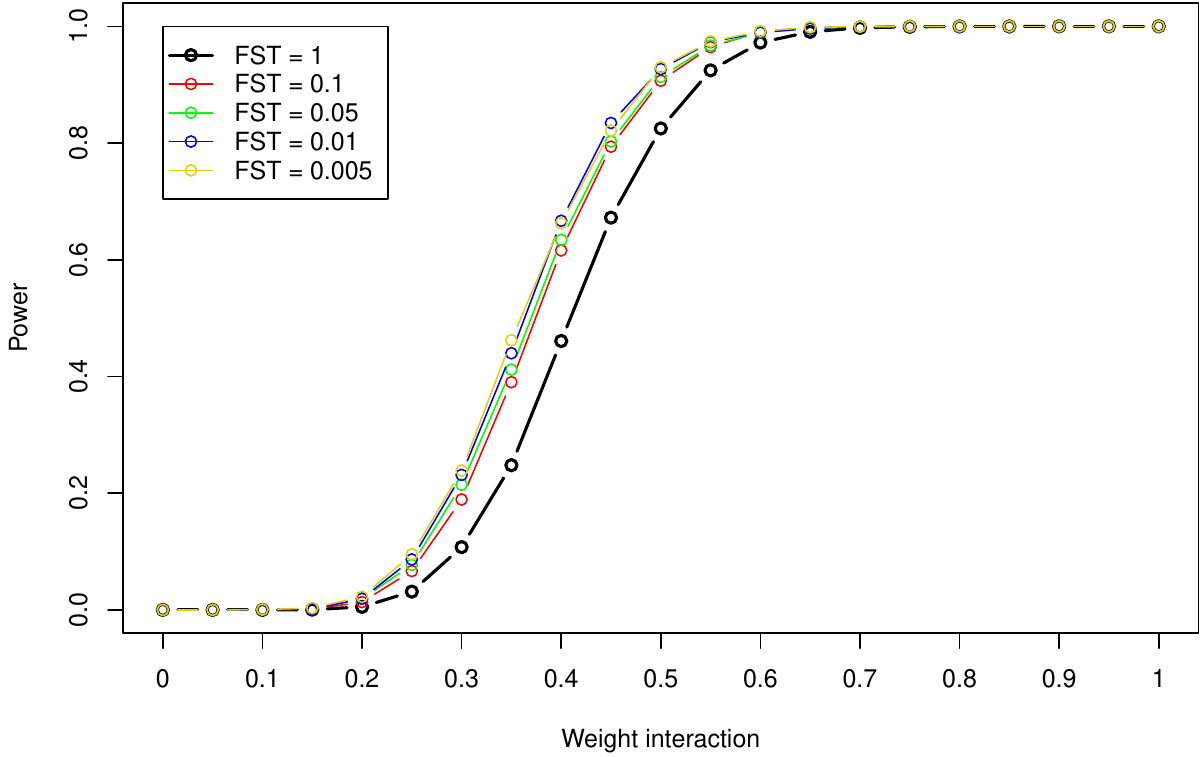}\;
\includegraphics[height=0.32\textwidth,width=0.32\textwidth]{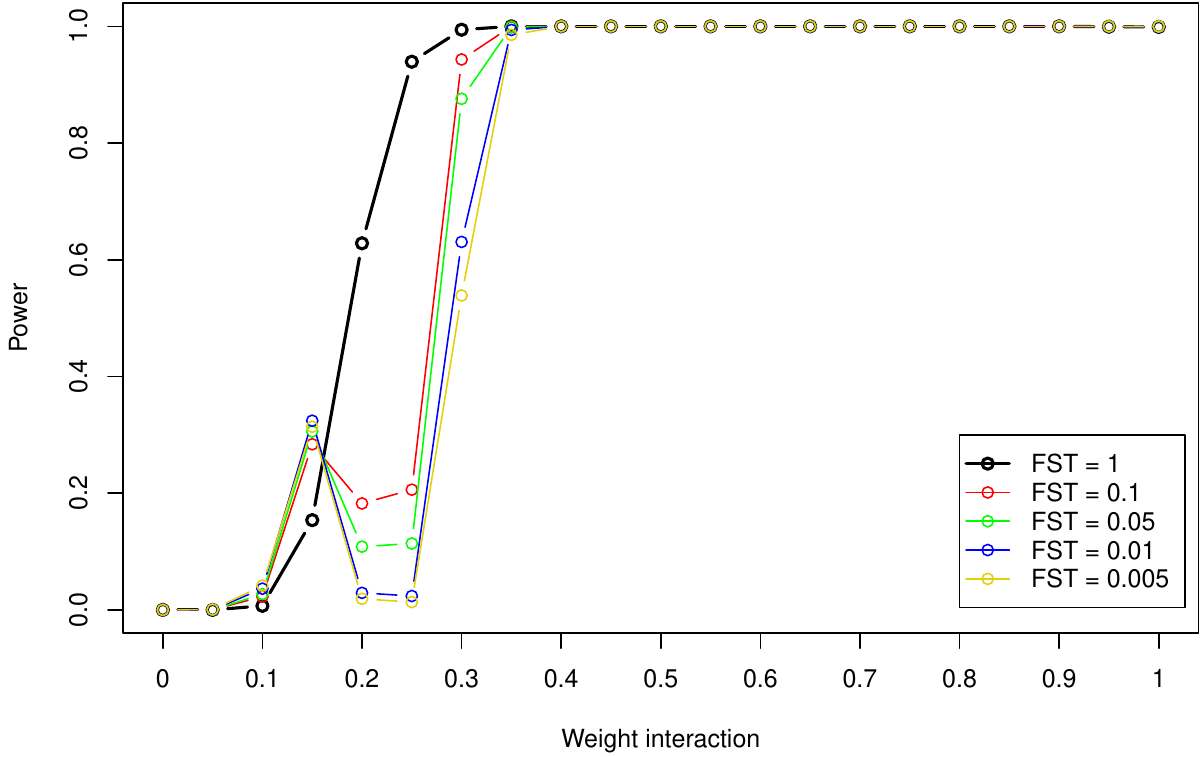}
\caption{Power as a function of the interaction effect for different first stage threshold (FST), with correlated markers. First row: linear regression model; second row: Cox PH model. First column: no main effects; Second column: both main effects equal 0.5; Third colum: both main effects equal -0.5. (Note the different scales on the x-axes.) }
\label{fig:power_correlatedmarkers}
\end{center}
\end{figure}

\subsection{Justification of using a marginal model for screening}
\label{sub:marginal}

In the first stage, markers are screened for interaction one by one by testing for association between the marker and the outcome variable marginally. Even if the marker is associated with the outcome via an interaction, and not marginally, the test in this misspecified marginal model has power to detect markers that are involved in an interaction. Mathematically this can be shown in the linear model in which explicit expressions for the estimators exist, but this is more complex to show in the logistic, Poisson and Cox PH models. As an alternative we performed a simulation study. We simulated marker data (the covariates $x_1, x_2$) by sampling them independently from a binomial distribution with parameters $m=2$ and probability 0.3 (to reflect the marker data). Next we simulated outcome data for the different models that include the interaction $x_1 x_2$ only (no main effects). The simulation procedure is as described before. Simulations are performed for interaction effects $\beta$ running from -1 to 1 in small steps of 0.002. For every value of $\beta$ marker and outcome data of 2000 individuals are sampled. For every data set the Wald test-statistics is computed for testing for association between $x_1$ and the outcome variable in a marginal regression model that includes the variable $x_1$ only. In Figure \ref{fig:marginal} the value of the test-statistic is plotted against the value of $\beta$ (the interaction-effect). In all three plots it can be seen that the test-statistic is close to zero when $\beta$ is close to zero and the value of the test-statistic deviates more from zero (and the smaller the accompanying p-values) the more $\beta$ differs from zero. This illustrates that the test in the first stage has power to detect an interaction effect and justifies the choice of the test in the first stage.

\begin{figure}[h]
\begin{center}
\includegraphics[height=0.32\textwidth,width=0.32\textwidth]{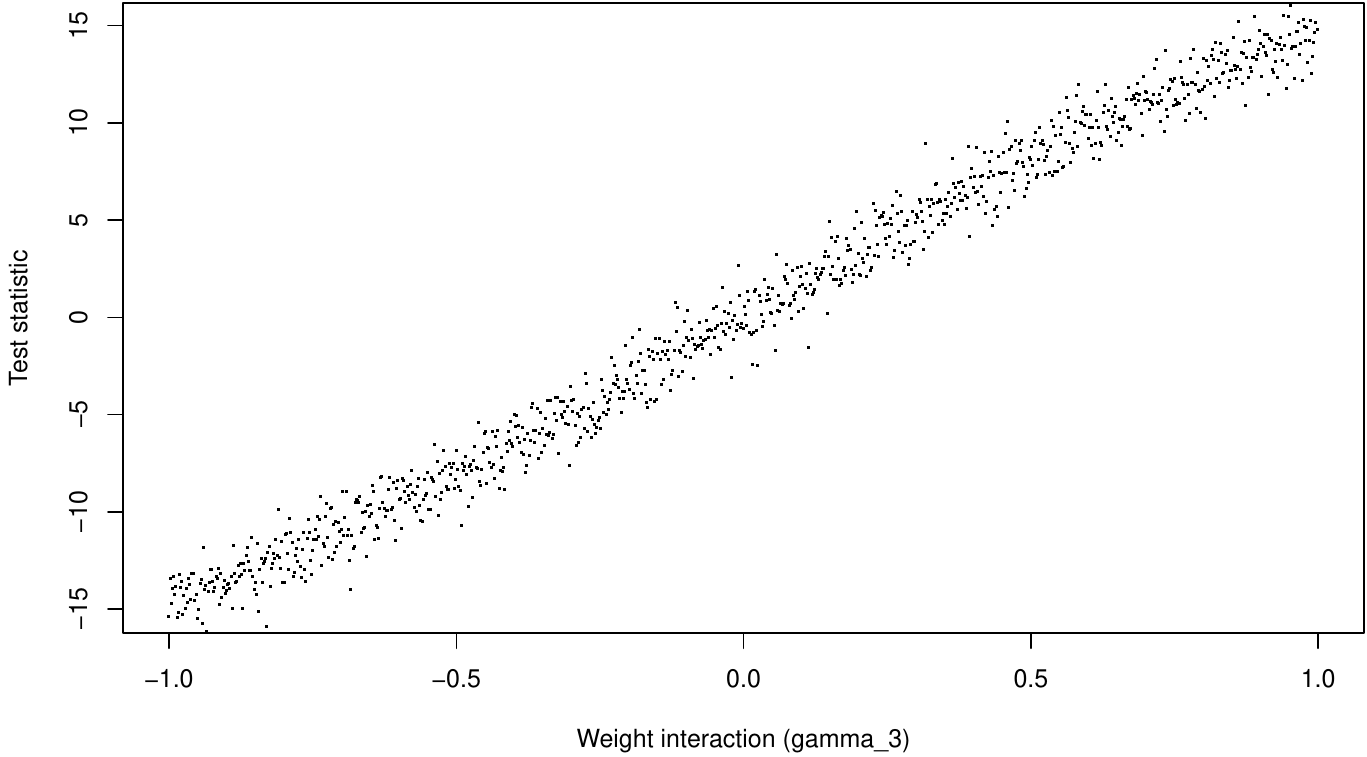} \;\;\;
\includegraphics[height=0.32\textwidth,width=0.32\textwidth]{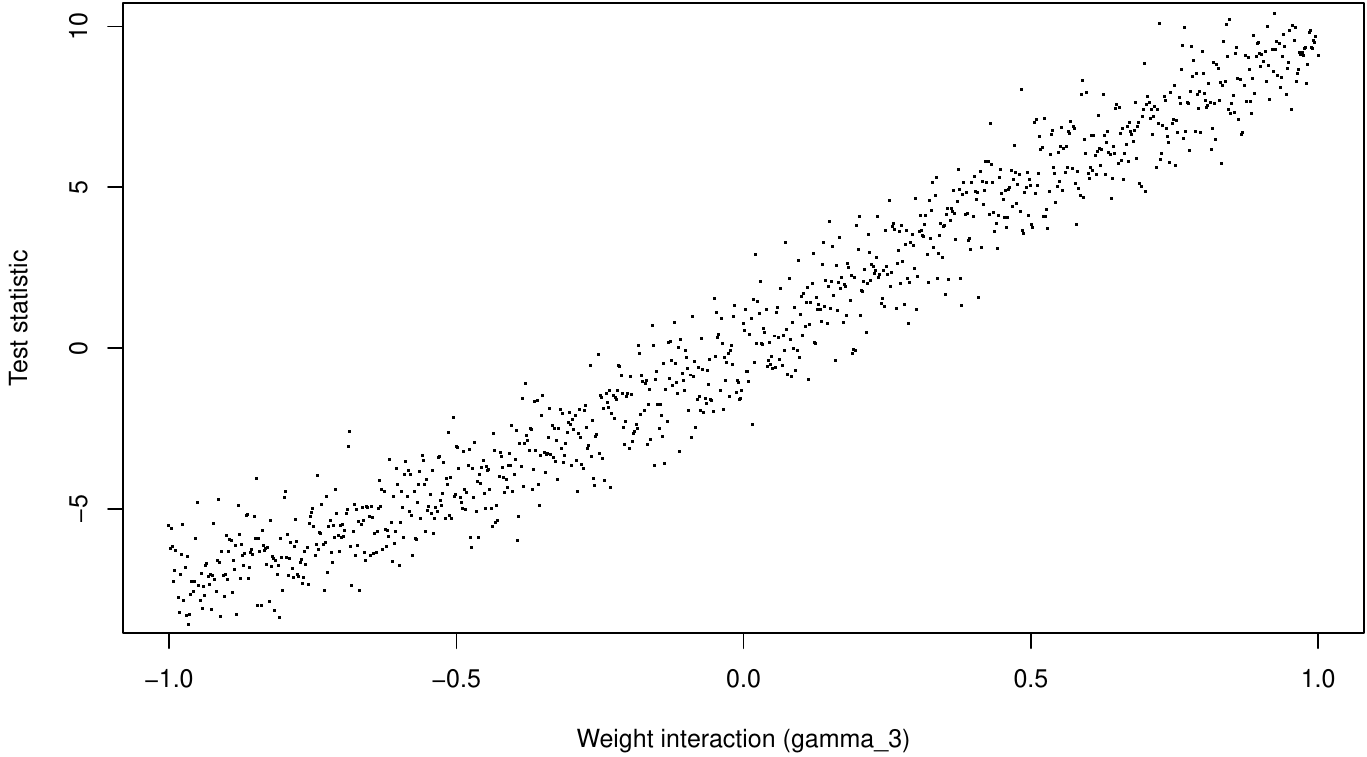}
\includegraphics[height=0.32\textwidth,width=0.32\textwidth]{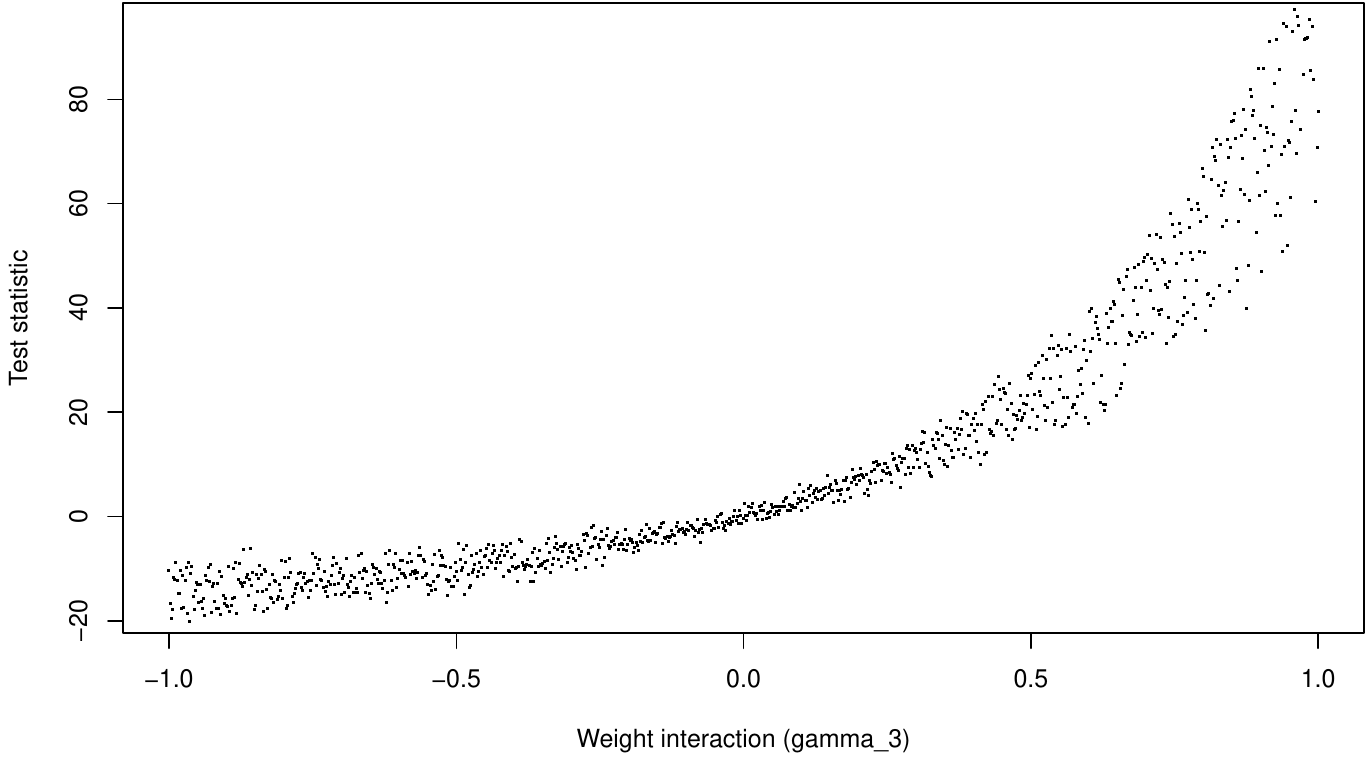} \;\;\;
\caption{Scatter plot for the value of the Wald test-statistics of the main effect in the marginal model against the value of the interaction effect in the full model. Left: linear regression model, middle: Cox PH model, right: Poisson regression model.}
\label{fig:marginal}
\end{center}
\end{figure}

\subsection{Data analysis}
Non-muscle invasive bladder cancer (NMIBC) often recurs after initial diagnosis. To investigate which markers are associated with the time from initial diagnosis to recurrence, data from 1451 patients with NMIBC were analyzed (Hof et al (2023)\nocite{Hof}). For all patients, genotypes of 670.542 markers are available for analysis. Of the 1451 patients, 591 had a recurrence during their follow-up time. In this paper the aim is to identify marker-marker interactions that are associated with the time to recurrence after initial diagnosis with the proposed two-stage method. The significance threshold in the first stage was chosen to be equal to 0.0001, later also 0.001 was considered. When using 0.0001 none of the marker pairs was significant in the second stage. For a first stage significance threshold of 0.001, the pair rs4415228 with rs2567720 was the only combination that was significant after Bonferroni correction, the corrected p-value was equal to 0.026. We have not found any relevant biological meaning of these SNPs. As far as we know, these SNPs have not been found in other bladder cancer studies.

\section{Discussion}
\label{sec: discussion}
In the paper we described the methodology for a two stage testing procedure with independent stages for detecting interactions in GLMs and the Cox regression model. In the first stage all covariates are tested for association with the outcome of interest in a model that includes the main effect of a single covariate only. If the corresponding p-value is below a pre-specified threshold, the covariate is passed on to the second stage. In the second stage all pairs of covariates that survived the first stage are tested for interaction.  

The significance threshold in the first stage of the two-stage testing procedure can be chosen freely. In the most extreme case the threshold is taken equal to 1 for some covariates which implies that these covariates are passed on to the second stage without being tested. In, for instance, in Dai et al (2012) the researchers are interested in testing interaction between genetic markers and  environmental factors.\nocite{Dai} They chose to screen the genetic markers only in the first stage and to pass on the environmental factors to the second stage directly; their threshold was taken equal to 1. For this specific setting some of the proofs can be simplified. 

In the simulation study we considered the setting with independent and with correlated marker genotypes. For independent markers the type I errors were close to the significance level of the test. In the correlated setting the multiple testing correction was too conservative; the type I error was clearly below the threshold. A less conservative multiple testing correction may lead to more power. We leave this for a next project. 

Although the two-stage testing strategy was proposed as a testing procedure for detecting interactions, it may also detect main-effects via the interaction term. By the choice of the test-statistics in the verification stage, the testing procedure may not be very powerful for detecting main effects. This is not seen as a weakness of the procedure as the focus is on detecting interaction effects and main effects are by-catch. 

In the simulation studies and the data analysis the calculations of the p-values are based on the asymptotic normality of the test-statistic under the null hypothesis. The normality approximations may not be very accurate for extreme large and small (negative) values of the test-statistics. Since the significance thresholds are usually very low, the low p-values in the second stage can be recalculated with the Firth approximation to increase the accuracy (see Firth (1993)).\nocite{Firth} This will increases the calculation time, but this does not need to be a computational burden as only a fraction of all possible interactions are tested in the second stage. In the analysis of the bladder cancer data, recalculation did not change the results. 

As an application of the methodology, detection of interactions between genetic markers or between genetic markers and environmental factors is given. However, its applicability is much wider, it can be applied in any high dimensional setting in which one is interested in detecting interaction effects.

\subsubsection*{Acknowledgement}
We would like to thank the Nijmegen Bladder Cancer Study project	team of the	Department for IQ Health of the Radboudumc in Nijmegen for giving us access to their database. This work made use of the Dutch national e-infrastructure with the support of the SURF Cooperative using grant no. EINF-1968.

\subsubsection*{Software}
An R package ``twostageGWASsurvival'' for doing the analysis have been developed and is available at github: {\tt 
https://github.com/luc-van-schijndel/twostageGWASsurvival}. The package implements a two stage method for the Cox model for searching for interaction effects within a high-dimensional setting where the covariates outnumber the subjects.

\appendix
Outline of the appendix:\\[-18pt] 
\begin{itemize}
\item[-] Appendix A: Pairwise independence in GLMs
\item[-] Appendix B: Complete independence in GLMs
\item[-] Appendix C: FWER control
\item[-] Appendix D: Pairwise independence in the Cox PH regression model
\item[-] Appendix E: Power plots for the Poisson regression model
\end{itemize}

\section*{Appendix A: Pairwise independence in GLMs}

\noindent
{\bf Theorem \ref{theorem:pairwise}.}
Suppose that $Y_1|{\bX}_1, \ldots, Y_n|{\bX}_n$ is a sample from the generalized linear model with a canonical link function, $M_{true}$.
Let $\hat{\btheta}_1$ and $\hat{\btheta}_2$ denote the maximum likelihood estimators of ${\theta}_1^0$ and ${\btheta}_2^0$ in the nested models $M_1$ and $M_2$, respectively. If model $M_2$ coincides with the true model $M_{true}$ (i.e., $p_2 = p_0$), then $\sqrt{n}(\hat\btheta_1-\btheta_1^0)$ and  $\sqrt{n}(\hat\btheta_2^\perp-\btheta_2^{0,\perp})$ are asymptotically independent.

\begin{proof}
The conditional density $f_{\btheta_0}$ of $Y|\bX$ from the true model $M_{true}$ is assumed to belong to a GLM with a canonical link function and, therefore, can be written in the form:
\begin{align*}
f_{\btheta_0}(Y|\bX)= g(Y)\exp\big(V(Y)^t \bZ \btheta_0 - A(\bZ \btheta_0)\big), 
\end{align*}
where $\bZ$ is a function of the covariates $\bX$.
The log-likelihood function (divided by $n$) for the true model $M_{true}$ equals:
\begin{align*}
    \frac{1}{n}\ell(\btheta_0) \;\propto\; \frac{1}{n} \sum_{i=1}^n \log \big\{f_{\btheta_0}(Y_i|{\bX}_i)\big\}  
    \;=\;  \frac{1}{n} \sum_{i=1}^n \big\{V(Y_i)^t \bZ_i\btheta_0 - A(\bZ_i\btheta_0)\big\} + \frac{1}{n} \sum_{i=1}^n \log \big\{g(Y_i)\big\}.  
\end{align*}
On the far right side, only the first term depends on the parameter of interest $\btheta_0$. Therefore, we can leave out the second term from the calculations. Next, we consider the average (conditional) log-likelihood for the submodel $M_1$:
\begin{align}
    \frac{1}{n}\ell_1(\btheta_1)  \; :=\; \frac{1}{n}\ell(E_1 \btheta_1)
     \;\propto\; \frac{1}{n} \sum_{i=1}^n\big(V(Y_i)^t\bZ_iE_1 \btheta_1 - A(\bZ_i E_1 \btheta_1)\big).
    \label{emp llh}
\end{align}
By the strong law of large numbers $\ell_1(\btheta_1)/n$ converges pointwise a.s.\ to its expectation: $\EE_{Y,\bX}\{V(Y)^t\bZ E_1 \btheta_1 - A(\bZ E_1\btheta_1)\}$. Let $\btheta_1^0 \in \mathbb{R}^{q_1}$ be the maximizer of this limit. It is known that the log-likelihood of an exponential family is concave, therefore this maximizer is unique  and a zero of its derivative with respect to $\btheta_1$. For $DA(.)$ defined as the gradient of $A(.)$ it holds that:
\begin{align}
    0  &=\; \EE_{Y,\bX}\Big\{V(Y)^t\bZ E_1 - DA(\bZ E_1\btheta_1^0)\bZ E_1\Big\} \nonumber \\[6pt]
    &= \; \EE_{\bX}\EE_{Y|\bX}\Big\{V(Y)^t\bZ E_1 - DA(\bZ E_1\btheta_1^0)\bZ E_1\Big\}
     \;=\; \EE_{\bX}\Big\{DA(\bZ\btheta^0_0)\bZ E_1 - DA(\bZ E_1\btheta_1^0)\bZ E_1\Big\}, 
     \label{llh id}
\end{align}
with $\btheta^0_0$ the true parameter value in model $M_{true}$. The last equality follows from the exponential family identity: $\EE_{Y|\bX}(V(Y)) = DA(\bZ\btheta^0_0)^t$.
Note that if $\btheta^0_0 = E_1\btheta_1^0$, the right hand side is trivially equal to zero.

By the argmax continuous mapping theorem and implications (see e.g., Van der Vaart and Wellner (1996): Subsection 3.2\nocite{vdV1996}) the maximum likelihood estimator for $\theta_1$ in model $M_1$, the maximizer of the log-likelihood (\ref{emp llh}), denoted as $\hat\btheta_1$, converges in probability to $\btheta_1^0$. For any finite sample define $\hat h_1:=\hat\btheta_1-\btheta^0_1\in \mathbb{R}^{q_1}$ and $\tilde{\ell}_1(\hat h_1):=\ell_1(\btheta_1^0+\hat h_1)$. Then, by a Taylor expansion of $A(\bZ_i E_1(\btheta_1^0 + \hat h_1))$ with respect to $\hat h_1$ near zero:
\begin{align}
    \frac{1}{n}\tilde{\ell}_1(\hat h_1)  &= \frac{1}{n}\sum_{i=1}^n \Big(V(Y_i)^t\bZ_i E_1(\btheta_1^0 + \hat h_1) - A(\bZ_i E_1(\btheta_1^0 + \hat h_1))\Big)\nonumber\\[4pt]
    & =\frac{1}{n}\sum_{i=1}^n \Big(V(Y_i)^t \bZ_i E_1\theta_1^0 + V(Y_i)^t \bZ_i E_1 \hat h_1 - A(\bZ_i E_1\theta_1^0)- DA(\bZ_i E_1\theta_1^0)  \bZ_i E_1 \hat h_1 \nonumber\\[4pt]
    & ~~~~\qquad - \frac{1}{2} \hat h_1^t E_1^t \bZ_i^t D^2A(\bZ_i E_1\theta_1^0) \bZ_i E_1 \hat h_1\Big)+O_p\big(\|\hat h_1\|^3\big), 
\label{TaylorExp}
\end{align}
where $DA$ and $D^2A$ are the first and second derivative of $A$ or the gradient and the Hessian matrix.

The term $O_p(\|\hat h_1\|^3)$ will converge to zero in probability. The first and the third term within the summation in (\ref{TaylorExp}) do not depend on $\hat h_1$ and their sum is written as $K_n$. By the law of large numbers, the fifth term equals
\begin{align*}
\frac{1}{2n}\sum_{i=1}^n \hat{h}_1^t E_1^t \bZ_i^t D^2A(\bZ_i E_1\theta_1^0) \bZ_i E_1 \hat{h}_1 = \tfrac{1}{2}\hat{h}_1^t E_1^t I_1 E_1 \hat{h}_1 + o_p\big(\|\hat{h}_1\|^2\big),
\end{align*}
with $I_1 := \EE_\bX{\bZ^t D^2A(\bZ E_1\btheta_1^0)\bZ} \in \mathbb{R}^{q_0\times q_0}$. The sum of the second and fourth term in (\ref{TaylorExp}) equals
\begin{equation*}
    \frac{1}{n}\sum_{i=1}^n \Big(V(Y_i)^t\bZ_i - DA(\bZ_i E_1\theta_1^0)\bZ_i\Big)E_1\hat{h}_1 = \frac{1}{\sqrt{n}} \bW_1^t E_1 \hat{h}_1 + o_p\big(\|n^{-1/2}\hat h_1\|\big),
\end{equation*}
by the central limit theorem, for $\bW_1 \in \mathbb{R}^{q_0}$ normally distributed. 

The same derivation holds for submodel $M_2$, yielding the normally distributed variable $\bW_2 \in \mathbb{R}^{q_0}$. Jointly, 
\begin{equation*}
    \mathbb{R}^{1\times 2q_0} \ni {\bW_1\choose \bW_2} = \lim_{n\to\infty} \frac{1}{\sqrt{n}}\sum_{i=1}^n {\bZ_i^t V(Y_i) - \bZ_i^t DA(\bZ_i E_1\theta_1^0)^t \choose \bZ_i^t V(Y_i) - \bZ_i^t DA(Z_iE_2\theta_2^0)^t}
\end{equation*}
has a normal distribution with mean zero by the equality in (\ref{llh id}). Next we compute its covariance-matrix. We start with the variance of $W_1$:
\begin{align*}
    \EE_{Y,\bX}{\bW_1\bW_1^t} & = \EE_{Y,\bX}\Big\{\Big(\bZ^tV(Y) - \bZ^tDA(ZE_1\theta_1^0)^t\Big)\Big(V(Y)^t\bZ - DA(\bZ E_1\theta_1^0)\Big)\Big\}\\[6pt]
    & = \EE_{Y,\bX}\Big\{\Big(\br{\bZ^tV(Y) - \bZ^tDA(\bZ\theta_0^0)^t} - \br{\bZ^tDA(\bZ E_1\theta_1^0)^t - \bZ^tDA(Z\theta_0^0)^t}\Big)\\[6pt]
    & \qquad  \qquad \times \Big({\br{V(Y)^t\bZ - DA(\bZ\theta_0^0)\bZ} - \br{DA(\bZ E_1\theta_1^0)Z - DA(\bZ\theta_0^0)Z}}\Big)\Big\}\\[6pt]
    & = \EE_{Y,\bX}\Big\{\Big(\bZ^tV(Y) - \bZ^tDA(\bZ\theta_0^0)^t\Big)\Big(V(Y)^t\bZ - DA(\bZ\theta_0^0)Z\Big)\Big\}\\[6pt]
    & \qquad + \EE_{Y.\bX}\Big\{\Big(\bZ^tDA(\bZ E_1\theta_1^0)^t - Z^tDA(\bZ\theta_0^0)^t\Big)\Big(DA(ZE_1\theta_1^0)\bZ - DA(\bZ\theta_0^0)Z\Big)\Big\}\\[6pt]
    & \qquad  - \EE_{Y,\bX}\Big\{\Big(\bZ^tV(Y) - \bZ^tDA(\bZ\theta_0^0)^t\Big)\Big(DA(\bZ  E_1\theta_1^0)\bZ - DA(\bZ\theta_0^0)\Big)\Big\}\\[6pt]
    & \qquad - \EE_{Y,\bX}\Big\{\Big(\bZ^tDA(\bZ E_1\theta_1^0)^t - \bZ^tDA(\bZ\theta_0^0)^t\Big)\Big(V(Y)^t\bZ - DA(\bZ\theta_0^0)\Big)\Big\}.
\end{align*}
The first of the four terms can be further simplified, by applying the exponential family-identity, $\Cov_{Y|\bX}(V(Y)) = D^2A(Z\theta_0^0)$:
\begin{align}
    &\ \EE_{Y,\bX}\Big\{\Big(Z^tV(Y) - \bZ^tDA(Z\theta_0^0)^t\Big)\Big(V(Y)^tZ - DA(Z\theta_0^0)\bZ \Big)\Big\} \nonumber \\[6pt]
    &\qquad=\ \EE_X{Z^t\EE_{Y|\bX}\Big\{\Big(V(Y) - DA(\bZ\theta_0^0)^t\Big)\Big(V(Y) - DA(\bZ\theta_0^0)^t}\Big)^t\Big\} \nonumber \\[6pt]
    &\qquad=\ \EE_\bX\Big\{{Z^t \Cov_{Y|X}\big(V(Y)\big)}\Big\} \nonumber \\[6pt]
    &\qquad=\ \EE_\bX\Big\{Z^t D^2A(Z\theta_0^0) \bZ\Big\}\;=:\ I \in \mathbb{R}^{q_0\times q_0}.
 \label{I}   
\end{align}
The second term can not be simplified, we define it as:
\begin{align*}
    C_{11} &:= \EE_{X}\Big\{\Big(Z^tDA(ZE_1\theta_1^0)^t - Z^tDA(Z\theta_0^0)^t\Big)\Big(DA(ZE_1\theta_1^0)Z - DA(Z\theta_0^0)Z\Big)\Big\} \in \mathbb{R}^{q_0\times q_0}.
\end{align*}
The last two terms equal zero. This can be seen as follows: 
\begin{align*}
    &\  \EE_{Y,X}\Big\{\Big(Z^tV(Y) - Z^tDA(Z\theta_0^0)^t\Big)\Big(DA(ZE_1\theta_1^0)Z - DA(Z\theta_0^0)Z\Big)\Big\}\\[6pt]
    &\qquad=\ \EE_X \mathbb{E}_{Y|X}\Big\{\Big(Z^tV(Y) - Z^tDA(Z\theta_0^0)^t\Big)\Big(DA(ZE_1\theta_1^0)Z - DA(Z\theta_0^0)Z\Big)\Big\}\\[6pt]
    &\qquad= \ \EE_X\Big\{\Big(Z^t\EE_{Y|X}\big(V(Y)\big) - Z^tDA(Z\theta_0^0)^t\Big)\Big(DA(ZE_1\theta_1^0)Z - DA(Z\theta_0^0)Z\Big)\Big\} \;=\; 0,
\end{align*}
since $\EE_{Y|X}\{V(Y)\} = DA(Z\theta_0^0)^t$. For the fourth term similar reasoning holds.
Summarized
\begin{align*}
\EE_{Y,X}\big\{W_1W_1^t\big\} = I + C_{11}.
\end{align*} 
Similarly, we find $\EE_{X,Y}\{W_2W_2^t\} = I + C_{22}$, with $C_{22}$ defined as
\begin{align*}
C_{22} &:= \EE_{X}\Big\{\Big(Z^tDA(ZE_2\theta_2^0)^t - Z^t DA(Z\theta_0^0)^t\Big)\Big(DA(ZE_2\theta_2^0)Z - DA(Z\theta_0^0)Z\Big)\Big\} \in \mathbb{R}^{q_0\times q_0}.
\end{align*}
The covariance of $W_1$ and $W_2$:
\begin{align*}
    \EE_{Y,X}{W_1W_2^t} & = \EE_{Y,X}\Big\{\Big(Z^tV(Y) - Z^tDA(ZE_1\theta_1^0)^t\Big)\Big(V(Y)^tZ-DA(ZE_2\theta_2^0)Z\Big)\Big\}\\[6pt]
    & = \EE_{Y,X}\Big\{\Big(\big(Z^tV(Y) - Z^tDA(Z\theta_0^0)^t\big) - \big(Z^tDA(ZE_1\theta_1^0)^t - Z^tDA(Z\theta_0^0)^t\big)\Big)\\[6pt]
    & \qquad \times\Big(\big(V(Y)^tZ - DA(Z\theta_0^0)Z\big)-\big(DA(ZE_2\theta_2^0)Z - DA(Z\theta_0^0)Z\big)\Big)\Big\}\\[6pt]
    & = \EE_{X}\Big\{Z^t D^2A(Z\theta_0^0) Z\Big\} 
     \;+\; \EE_{X}\Big\{\Big(Z^tDA(ZE_1\theta_1^0)^t - Z^tDA(Z\theta_0^0)^t\Big)\Big(DA(ZE_2\theta_2^0)Z-DA(Z\theta_0^0)Z\Big)\Big\}\\[6pt]
     &= I + C_{12}
\end{align*}
with (by similar reasoning)
\begin{align*}
    C_{12} := \EE_{X}\Big\{\Big(Z^tDA(ZE_1\theta_1^0)^t - Z^tDA(Z\theta_0^0)^t\Big)\Big(DA(ZE_2\theta_2^0)Z-DA(Z\theta_0^0)Z\Big)\Big\} \in \mathbb{R}^{q_0\times q_0},
\end{align*}
and $C_{12} = {C_{12}}^t$. 
Now, due to the combined vector $(W_1^t, W_2^t)^t$ having expectation zero, we find its covariance matrix:
\begin{equation*}
    \Cov_{Y,X}{W_1\choose W_2} = \begin{pmatrix}  I + C_{11} &~~~~ I + C_{21}\\
    I + C_{12} &~~~~ I + C_{22} \end{pmatrix} \in \mathbb{R}^{2q_0\times 2q_0}.
\end{equation*}
Note that if the submodel $M_2$ coincides with model $M_{true}$, then $E_2\theta_2^0 = \theta_0^0$ and $C_{21} = C_{12} = C_{22} = \mathbf{0}\in \mathbb{R}^{q_0\times q_0}$. 

Summarized
\begin{align}
\label{limit in d average loglikelihood proof}
    \frac{1}{n}\tilde{\ell}(\hat{h}_1) = K_n + \frac{1}{\sqrt{n}}W_1^tE_1\hat{h}_1  - \frac{1}{2}\hat{h}_1^t E_1^t I_1 E_1 \hat{h}_1 + o_p\big(\|n^{-1/2}\hat h_1\|\big) + o_p\big(\|\hat h_1\|^2\big),
\end{align}
In a similar way $n^{-1}\tilde\ell(\hat h_2)$ can be written as a sum of asymptotic terms, with $I_2: = \EE_X{Z^t D^2A(ZE_2\theta_2^0)Z} \in \mathbb{R}^{q_0 \times q_0}$. Maximizing $n^{-1}\tilde{\ell}(h_1)$ with respect to $h_1$ to obtain an expression for $\hat h_1$ 
gives for $\hat h_1$, and in a similar way for $\hat h_2$:
\begin{align*}
\hat h_1 & = \frac{1}{\sqrt{n}}\br{E_1^tI_1E_1}^{-1} E_1^t W_1 \in \mathbb{R}^{q_1}\\[4pt]
\hat h_2 & = \frac{1}{\sqrt{n}}\br{E_2^tI_2E_2}^{-1} E_2^t W_2 \in \mathbb{R}^{q_2}.
\end{align*}
In order to analyze the asymptotic covariance of $\hat h_1$ and $\hat h_2$ orthogonal to the submodel $M_1$, this is found by $(Id_{q_2} - E_2^tE_1E_1^tE_2)\hat h_2 \in \mathbb{R}^{q_2}$, where $Id_{q_2}$ denotes the $q_2$-dimensional identity matrix. Ultimately, the $q_1\times q_2$-covariance becomes:
\begin{align}
    \Cov_{Y,X}\Big(&\hat h_1, \br{Id_{q_2} - E_2^tE_1E_1^tE_2}\hat h_2\Big)\nonumber\\[4pt]
    & = \frac{1}{n}\br{E_1^tI_1E_1}^{-1}E_1^t \;\Cov_{Y,X}(W_1, W_2) E_2\br{E_2^tI_2E_2}^{-1} \br{Id_{q_2} - E_2^t E_1 E_1^t E_2}\nonumber\\[4pt]
    & = \frac{1}{n}\br{E_1^tI_1E_1}^{-1}E_1^t \br{I+C_{12}} E_2\br{E_2^tI_2E_2}^{-1} \br{Id_{q_2} - E_2^t E_1 E_1^t E_2}.
    \label{covariance theorem}
\end{align}
To complete the proof, we set $M_2 = M_{true}$, so $\theta_0^0 = E_2\theta_2^0$ which in turn implies $I_2 = I$ and $C_{21} = \mathbf{0}$. By the definitions of $E_1$ and $E_2$ as projection matrices:
Furthermore:
\begin{align}
    E_1^t & = E_1^t E_2 E_2^t \label{id 1}\\
    E_1^t & = E_1^t E_1 E_1^t \label{id 2},
\end{align}
where the equality in (\ref{id 1}) follows from the fact that $M_1$ is nested within $M_2$. Apply these to find:
\begin{align}
\label{covariance}
\Cov_{Y,X}\big(\sqrt{n}(\hat\btheta_1-\btheta_1^0), \sqrt{n}(\hat\btheta_2-\btheta_2^0)^\perp\big) &=
    n \; \Cov_{Y,X}\big(\hat h_1, \br{Id_{q_2} - E_2^tE_1E_1^tE_2}\hat h_2\big) \nonumber \\[4pt]
    & = \br{E_1^tI_1E_1}^{-1}E_1^tE_2\br{E_2^t \br{I_2+\mathbf{0}} E_2}\br{E_2^tI_2E_2}^{-1} \br{Id_{q_2} - E_2^t E_1 E_1^t E_2} \nonumber  \\[4pt]
    & = \br{E_1^tI_1E_1}^{-1}E_1^tE_2 \br{Id_{q_2} - E_2^t E_1 E_1^t E_2} \nonumber \\[4pt]
    & = \br{E_1^tI_1E_1}^{-1}\br{E_1^tE_2 - E_1^tE_2E_2^t E_1 E_1^t E_2} \nonumber \\[4pt]
    & = \br{E_1^tI_1E_1}^{-1}\br{E_1^tE_2 - E_1^tE_1E_1^t E_2} \nonumber  \\[4pt]
    & = \br{E_1^tI_1E_1}^{-1}\br{E_1^tE_2 - E_1^t E_2} \nonumber  \\[4pt]
    & = \mathbf{0},
\end{align}
where the second and fifth equalities follow from (\ref{id 1}), and the sixth equality follows from (\ref{id 2}). By the joint normality of $\sqrt{n}\; \hat h_1$ and $\sqrt{n}\;\hat h_2$, a zero correlation implies the required independence as stated in the theorem.  
\end{proof}

\bigskip

\noindent
{\bf Theorem \ref{theorem:pairwiseGaussian} (Pairwise independence in the linear regression model).} 
Suppose that $Y_1|{\bX}_1, \ldots, Y_n|{\bX}_n$ is a sample from the linear regression model $M_2$ in (\ref{M2}) with the identical link function and unknown variance $\sigma^2$. Then $\hat\beta_1$ and  $\hat\bbeta_2^\perp$ are independent.

\noindent
\begin{proof} 
The pair $(\hat\bbeta_1,\hat\bbeta_2)$ follows a multivariate normal distribution as $\hat\bbeta_1$ and $\hat\bbeta_2$ are linear combinations of $Y_1,\ldots, Y_n$. To proof independence between $\hat\beta_1$ and $\hat\bbeta_2^\perp$, we have to show that their covariance equals zero: $\Cov(\hat\beta_1,\hat\bbeta_2^\perp)=0$. To prove this, we compute the covariance $\Cov(\hat\bbeta_1,\hat\bbeta_2)$ and show that the elements of the covariance matrix that represent the covariance between $\hat\beta_1$ and $\hat\bbeta_2^\perp$ equal zero.    

Define $\bY=(Y_1,\ldots,Y_n)^t$ and $\bZ_1$ and $\bZ_2$ as the $p_1\times n$ and $p_2\times n$ matrices with the covariates of the $n$ observations in the models $M_1$ and $M_2$ in (\ref{MM1}) and (\ref{MM2}). Then, $\hat\bbeta_1 = (\bZ_1 \bZ_1^t)^{-1} \bZ_1 \bY$ and $\hat\bbeta_2 = (\bZ_2 \bZ_2^t)^{-1} \bZ_2 \bY$

\begin{align*}
\Cov\big(\hat\bbeta_1,\hat\bbeta_2\big) 
&= \Cov\big( (\bZ_1 \bZ_1^t)^{-1} \bZ_1 \bY, (\bZ_2 \bZ_2^t)^{-1} \bZ_2 \bY\big) \\[6pt]
&= \EE_{\bZ,\bY}\Big((\bZ_1 \bZ_1^T)^{-1} \bZ_1 \bY \bY^t \bZ_2^t (\bZ_2 \bZ_2^t)^{-1}\Big) - \EE_{\bZ,\bY}\Big((\bZ_1 \bZ_1^t)^{-1} \bZ_1 \bY\Big) \EE_{\bZ,\bY} \Big(\bY^t \bZ_2^t (\bZ_2 \bZ_2^t)^{-1}\Big) 
\end{align*}
with
\begin{align*}
\EE_{\bZ,\bY}\Big((\bZ_1 \bZ_1^t)^{-1} \bZ_1 \bY\Big) \;&= \EE_{\bZ}\Big((\bZ_1 \bZ_1^t)^{-1} \bZ_1\Big) \EE_{\bY|\bZ} \bY \;=\; \EE_{\bZ}\Big((\bZ_1 \bZ_1^t)^{-1} \bZ_1 \bZ_2^t \bbeta_2^0\Big) \\[6pt]
\EE_{\bZ,\bY} \Big(\bY^t \bZ_2^t (\bZ_2 \bZ_2^t)^{-1}\Big) \;&= \EE_{\bZ}\EE_{\bY|\bZ} \Big(\bY^t \bZ_2^t (\bZ_2 \bZ_2^t)^{-1}\Big) \;=\; \EE_{\bZ}\Big((\bbeta_2^0)^t \bZ_2 \bZ_2^t (\bZ_2 \bZ_2^t)^{-1}\Big) \;=\; (\bbeta_2^0)^t \\[6pt]
\EE_{\bZ,\bY}\Big((\bZ_1 \bZ_1^t)^{-1} \bZ_1 \bY\Big) \EE_{\bZ,\bY} \Big(\bY^t \bZ_2^t (\bZ_2 \bZ_2^t)^{-1}\Big) \;&=\; \EE_{\bZ}\Big((\bZ_1 \bZ_1^t)^{-1} \bZ_1 \bZ_2^t \bbeta_2^0 (\bbeta_2^0)^t\Big)
\end{align*}
and
\begin{align*}
\EE_{\bZ,\bY}\Big((\bZ_1 \bZ_1^t)^{-1} \bZ_1 \bY \bY^t \bZ_2^t (\bZ_2 \bZ_2^t)^{-1})\Big) &=\; \EE_{\bZ}\Big((\bZ_1 \bZ_1^t)^{-1} \bZ_1\Big) \EE_{\bY|\bZ} \Big((\bY \bY^t) \bZ_2^t (\bZ_2 \bZ_2^t)^{-1}\Big) \\[6pt]
&=\; \EE_{\bZ}\Big((\bZ_1 \bZ_1^t)^{-1} \bZ_1 \big\{\sigma^2 I_n + (\bZ_2^t\bbeta_2^0)(\bZ_2^t\bbeta_2^0)^t\big\}\; \bZ_2^t (\bZ_2 \bZ_2^t)^{-1}\Big) \\[6pt]
&=\; \sigma^2 \EE_{\bZ}\Big((\bZ_1 \bZ_1^t)^{-1} \bZ_1 \bZ_2^t (\bZ_2 \bZ_2^t)^{-1}\Big) 
\;+\; \EE_{\bZ}\Big((\bZ_1 \bZ_1^t)^{-1} \bZ_1  \bZ_2^t\bbeta_2^0 (\bZ_2^t\bbeta_2^0)^t \bZ_2^t (\bZ_2 \bZ_2^t)^{-1}\Big) \\[6pt]
&=\; \sigma^2 \EE_{\bZ}\Big((\bZ_1 \bZ_1^t)^{-1} \bZ_1 \bZ_2^t (\bZ_2 \bZ_2^t)^{-1}\Big) 
\;+\; \EE_{\bZ}\Big((\bZ_1 \bZ_1^t)^{-1} \bZ_1  \bZ_2^t\bbeta_2^0 (\bbeta_2^0)^t\Big) 
\end{align*}
That gives
\begin{align*}
\Cov\big(\hat\bbeta_1,\hat\bbeta_2\big) &=
\sigma^2 \EE_{\bZ}\Big((\bZ_1 \bZ_1^t)^{-1} \bZ_1  \bZ_2^t (\bZ_2 \bZ_2^t)^{-1}\Big)
\end{align*}
Since model $M_1$ is nested in model $M_2$, we can write $\bZ_2^t=(\bZ_1^t | \bZ_3^t)$ with $\bZ_3$ the covariate rows which are in $\bZ_2$ but not in $\bZ_1$. This gives $\bZ_1 \bZ_2^t = \bZ_1 (\bZ_1^t |  \bZ_3^t)= (\bZ_1 \bZ_1^t | \bZ_1 \bZ_3^t)$. And 
\begin{align*}
(\bZ_1 \bZ_1^t)^{-1} \bZ_1 \bZ_2^t (\bZ_2 \bZ_2^t)^{-1}\;=\; (\bZ_1 \bZ_1^t)^{-1}(\bZ_1 \bZ_1^t | \bZ_1 \bZ_3^t)(\bZ_2 \bZ_2^t)^{-1} \;=\; (I_{p_1} \; | \; (\bZ_1 \bZ_1^t)^{-1} \bZ_1 \bZ_3^t)(\bZ_2 \bZ_2^t)^{-1},
\end{align*}
with $I_{p_1}$ the $p_1\times p_1$ identity matrix.
Write
\begin{align*}
\bZ_2 \bZ_2^t \;=\; 
\begin{pmatrix}
\bZ_1 \bZ_1^t  &~~  \bZ_1 \bZ_3^t \\
\bZ_3 \bZ_1^t  &~~  \bZ_3 \bZ_3^t
\end{pmatrix} 
\end{align*}
Then, to shorten notation in the next display, let $\bW_{11}:=\bZ_1 \bZ_1^t$, $\bW_{33}:=\bZ_3 \bZ_3^t$ and $\bW_{13}:=\bZ_1 \bZ_3^t$  
\begin{align*}
\begin{pmatrix}
\bZ_1 \bZ_1^t  &~~  \bZ_1 \bZ_3^t \\
\bZ_3 \bZ_1^t  &~~  \bZ_3 \bZ_3^t
\end{pmatrix} ^{-1} &=\; 
\begin{pmatrix}
\bW_{11}   &~~  \bW_{13} \\
\bW_{13}^t  &~~  \bW_{33} 
\end{pmatrix} ^{-1}  \\[6pt]
&=\begin{pmatrix}
\bW_{11}^{-1} + \bW_{11}^{-1} \bW_{13} (\bW_{33} - \bW_{13}^t \bW_{11}^{-1} \bW_{13})^{-1} \bW_{13}^t \bW_{11}^{-1} & ~~~ -\bW_{11}^{-1} \bW_{13} (\bW_{33} - \bW_{13}^t \bW_{11}^{-1} \bW_{13})^{-1}\\
-(\bW_{33}-\bW_{13}^t \bW_{11}^{-1} \bW_{13})^{-1} \bW_{13}^t \bW_{11}^{-1} & (\bW_{33} - \bW_{13}^t \bW_{11}^{-1} \bW_{13})^{-1}
\end{pmatrix}
\end{align*}
We need to show that the last block column of the matrix
\begin{align*}
(I_{p_1} \; | \; \bW_{11}^{-1} \bW_{13})
&\begin{pmatrix}
\bW_{11}  &~~  \bW_{13} \\
\bW_{13}^t  &~~  \bW_{33}
\end{pmatrix} ^{-1}
\end{align*}
equals zero. That block equals
\begin{align*}
-\bW_{11}^{-1}\bW_{13}(\bW_{33}-\bW_{13}^t\bW_{11}^{-1}\bW_{13})^{-1} + \bW_{11}^{-1} \bW_{13} (\bW_{33}-\bW_{13}^t\bW_{11}^{-1}\bW_{13})^{-1}  = {\bf 0}
\end{align*}
\end{proof}

\section*{Appendix B: Mixed pairwise independence in GLMs}
\label{App: complete independence}
In this Appendix we will prove Theorem \ref{theorem:full}.  If in Theorem \ref{theorem:full}, $j=k$ or $j=\ell$ the stated independence is pairwise independence and follows from Result \ref{theorem:pairwise}. Here we assume that $j\neq k$ and $j\neq \ell$. First, we give the proof in a more general setting. After the proof, the theorem is translated to the setting we are interested in. Define the models 
\begin{align}
M_1:\qquad h\big(\mathbb{E}(Y|\bX,\bbeta_1,\bnu_1)\big) &= \beta_{1,1} x_{1} + \beta_{1,2} x_{2} + \ldots + \beta_{1,p_1} x_{p_1} \label{FullM11}\\
M_2:\qquad h\big(\mathbb{E}(Y|\bX,\bbeta_2,\bnu_2)\big) &= \beta_{2,1} x_{1} + \beta_{2,2} x_{2} + \hspace{2.0cm} \beta_{2,p_1+1} x_{p_1+1} +  \ldots + \beta_{2,p_1+p_2} x_{p_1+p_2} \label{FullM22}\\
M_3:\qquad h\big(\mathbb{E}(Y|\bX,\bbeta_2,\bnu_3)\big) &= \beta_{3,1} x_{1} + \beta_{3,2} x_{2} + \ldots + \beta_{3,p_1} x_{p_1} +  \beta_{3,p_1+1} x_{p_1+1} + \ldots + \beta_{3,p_1+p_2} x_{p_1+p_2} \label{FullM33}\\
M_{true}:\qquad h\big(\mathbb{E}(Y|\bX,\bbeta_0,\bnu_0)\big) &= \beta_{0,1} x_{1} + \beta_{0,2} x_{2} \label{FullM00}
\end{align}
where the models $M_1, M_2$ and $M_{true}$ are defined as before. 
The models $M_1, M_2$ and $M_{true}$ are nested in model $M_3$. Model $M_{true}$ is the true model with only the intercept and the ``fixed'' covariates.
Define the vectors $\hat\btheta_i$ as the maximum likelihood estimator of the full parameter vector in model $M_i$, and $\btheta_0$ as the true parameter vector (so the regression parameters and the nuisance parameters). The dimension of the parameter spaces are denoted as $q_1, q_2, q_3$ and $q_0$.  

The parameter spaces for the models $M_1$ and $M_2$ are nested in model $M_3$. To project a parameter from the parameter space for $M_1$ or $M_2$ to the parameter space for model $M_3$ we define the projection matrices (as before) 
\begin{align*}
    E_1 & = \br{e_1|\ldots |e_{q_1}} \in \mathbb{R}^{q_3\times q_1}\\
    E_2 & = \br{e_1|\ldots |e_{q_2}} \in \mathbb{R}^{q_3\times q_2},
\end{align*}
in such a way that $\{e_1, \ldots, e_{q_1}\}$ forms an orthonormal basis for the parameter space for $M_1$ and similarly for $E_2$. The matrices represent linear functions from the parameter spaces for $M_1$ and $M_2$ to $M_3$.

\medskip

\begin{lemma}[{\bf Mixed pairwise independence in GLMs}]
\label{Lemma:full}
Suppose that $Y_1|{\bX}_1, \ldots, Y_n|{\bX}_n$ is a sample from the generalized linear model $M_{true}$ in (\ref{FullM00}) with a canonical link function. Let $\hat{\btheta}_1$ and $\hat{\btheta}_2$ be the maximum likelihood estimators of ${\theta}_1$ and ${\btheta}_2$ in the models $M_1$ and $M_2$ in (\ref{FullM11}) and (\ref{FullM22}), respectively. If $E_2 E_2^t \EE_X{Z^t \Cov_{Y|X}(V(Y)) Z} E_2 = \EE_X{Z^t \Cov_{Y|X}(V(Y)) Z} E_2$, where $Z$ and $V(Y)$ are defined as in the expression of the exponential family (\ref{expFamily}),  
then $\sqrt{n}(\hat\btheta_1-\btheta_1^0)$ and  $\sqrt{n}(\hat\btheta_2^0-\btheta_2^0)^\perp$ are asymptotically independent.
\end{lemma}

\begin{proof}
The statement in this lemma can be proven along the same lines as the proof of Theorem \ref{theorem:pairwuseGLM} up to the calculation of the covariance in (\ref{covariance}) at the end of the proof. In that calculation the equality $E_1^t = E_1^t E_2 E_2^t$ in (\ref{id 1}) is applied, but does not hold in the current setting as the model $M_1$ is not nested in $M_2$. However, if $E_2 E_2^t \EE_X{Z^t \Cov_{Y|X}(V(Y)) Z} E_2 = \EE_X{Z^t \Cov_{Y|X}(V(Y)) Z} E_2$, the covariance can be proven to be equal to zero in the same way. This is explained below. Remind that $I = \EE_X{Z^t \Cov_{Y|X}(V(Y)) Z}$. Under the null hypothesis that there are no main and no interaction effects, the true model coincides with the models $M_1$ and $M_2$. Then, the covariance of interest is given by
\begin{align*}
\Cov_{Y,X}\big(\sqrt{n}(\hat\btheta_1-\btheta_1^0), \sqrt{n}(\hat\btheta_2-\btheta_2^0)^\perp\big) &=
n \; \Cov_{Y,X}\big(\hat h_1, \br{Id_{q_2} - E_2^tE_1E_1^tE_2}\hat h_2\big) \nonumber \\[4pt]
&= \br{E_1^t I E_1}^{-1}E_1^t I E_2\br{E_2^t I E_2}^{-1} \br{Id_{q_2} - E_2^t E_1 E_1^t E_2}\nonumber
\end{align*}
Using the assumption that $I E_2 = E_2 E_2^t I E_2$ yields
\begin{align*}
\br{E_1^tI E_1}^{-1}E_1^t I E_2\br{E_2^tI E_2}^{-1} \br{Id_{q_2} - E_2^t E_1 E_1^t E_2} &= \br{E_1^tI E_1}^{-1}E_1^t E_2 E_2^t I E_2\br{E_2^tI E_2}^{-1} \br{Id_{q_2} - E_2^t E_1 E_1^t E_2}\\[4pt]
&= \br{E_1^tI E_1}^{-1}E_1^t E_2  \br{Id_{q_2} - E_2^t E_1 E_1^t E_2}\\[4pt]
&= \br{E_1^tI E_1}^{-1}E_1^t E_2 -  \br{E_1^tI E_1}^{-1}E_1^t E_2E_2^t E_1 E_1^t E_2\\[4pt]
&= \br{E_1^tI E_1}^{-1}E_1^t E_2 -  \br{E_1^tI E_1}^{-1}E_1^t E_1 E_1^t E_2\\[4pt]
&= \br{E_1^tI E_1}^{-1}E_1^t E_2 -  \br{E_1^tI E_1}^{-1}E_1^t E_2\\[4pt]
&= \mathbf{0}
\end{align*}
\end{proof}

The lemma above is proven under the assumptions that $E_2 E_2^t \EE_X{Z^t \Cov_{Y|X}(V(Y)) Z} E_2 = \EE_X{Z^t \Cov_{Y|X}(V(Y)) Z} E_2$. By straightforward calculations, it can be seen that for the linear, the Poisson, and the logistic regression model this assumption is satisfied if the covariates that are included in the two stage testing procedure (so not necessarily the fixed covariates) are mutually independent and independent of the fixed covariates, and the covariates are standardized to have mean zero. The assumption also holds if one of the covariates is a two-way interaction term of two covariates that are tested in the two-stage strategy and are included in the model as in the situation described in the main part of the paper.

\section*{Appendix C: FWER control}
Let $T_j^{S_1}, j=1,\ldots,p$ be the test-statistics for the $p$ covariates in the first stage, and let $\Gamma_j^{S_1}, j=1,\ldots,p$ be the corresponding rejection regions. Furthermore, let $T_{jk}^{S_2}$ be the test-statistic for testing interaction between the covariates $j$ and $k$ in the second stage, and let $\Gamma_{jk}^{S_2}$ be the corresponding rejection regions. In the second stage a null hypothesis is rejected if the p-value is below the threshold $\alpha$ divided by the number of tests that are performed in the second stage. This number is stochastic, and thus the rejection region $\Gamma_{jk}^{S_2}$ is stochastic as well.

\begin{lemma}[{\bf FWER control}]  
\label{lemma:FWER}
Suppose that, under the full null hypothesis, $(T_j^{S_1},T_k^{S_1})$ and $T_{jk}^{S_2}$ are asymptotically independent, and that, moreover, $T_j^{S_1}, j=1,\ldots,p$ are asymptotically independent. 
Then, under the full null hypothesis, 
the two-step procedure preserves the FWER at the significance level $\alpha$ in weak sense:
\begin{align*}
\rP\Big(\bigcup_{j,k=1:j<k}^{p,p} \big\{\big(T_j^{S_1} \in \Gamma_j^{S_1}\big) \cap \big(T_k^{S_1} \in \Gamma_k^{S_1}\big) \cap \big(T_{jk}^{S_2} \in \Gamma_{jk}^{S_2}\big)\big\}\Big) \leq \alpha.      
\end{align*}
\end{lemma}

\begin{proof}
Let $R_1$ be defined as the stochastic number of tests that are rejected in the first stage. Given that $R_1=r_1$ covariates were past on to the second stage, the set $\Gamma_{jk}^{S_2}(r_1)$ is defined as the non-stochastic rejection region for test-statistic $T^{S_2}_{jk}$. 
\begin{align*}
\lefteqn{\rP\Big(\bigcup_{j,k=1:j< k}^{p,p} \big\{\big(T_j^{S_1} \in \Gamma_j^{S_1}\big) \cap \big(T_k^{S_1} \in \Gamma_k^{S_1}\big) \cap \big(T_{jk}^{S_2} \in \Gamma_{jk}^{S_2}\big)\big\}\Big)}\\ 
&=\sum_{r_1=2}^p \rP(R_1=r_1)\rP\Big(\bigcup_{j,k=1:j< k}^{p,p} \big\{\big(T_j^{S_1} \in \Gamma_j^{S_1}\big) \cap \big(T_k^{S_1} \in \Gamma_k^{S_1}\big) \cap \big(T_{jk}^{S_2} \in \Gamma_{jk}^{S_2}\big)\big\}\big|R_1=r_1\Big)\\
&\leq \sum_{r_1=2}^p \rP(R_1=r_1) \sum_{j,k=1:j< k}^{p,p} \rP\big\{ \big(T_j^{S_1} \in \Gamma_j^{S_1}\big) \cap \big(T_k^{S_1} \in \Gamma_k^{S_1}\big) \cap \big(T_{jk}^{S_2} \in \Gamma_{kj}^{S_2}\big)\big|R_1=r_1\big\}\\
&= \sum_{r_1=2}^p \rP(R_1=r_1)\sum_{j,k=1:j< k}^{p,p} \rP\big\{\big(T_j^{S_1} \in \Gamma_j^{S_1}\big) \cap \big(T_k^{S_1} \in \Gamma_k^{S_1}\big)\big|R_1=r_1\big\}\rP\big\{T_{jk}^{S_2} \in \Gamma_{jk}^{S_2}\big| \big(T_j^{S_1} \in \Gamma_j^{S_1}\big) \cap \big(T_k^{S_1} \in \Gamma_k^{S_1}\big)\cap  \big(R_1=r_1\big)\big\}\\
&= \sum_{r_1=2}^p \rP(R_1=r_1)\sum_{j,k=1:j< k}^{p,p} \rP\big\{\big(T_j^{S_1} \in \Gamma_j^{S_1}\big) \cap \big(T_k^{S_1} \in \Gamma_k^{S_1}\big)\big|R_1=r_1\big\}\rP\big\{T_{jk}^{S_2} \in \Gamma_{jk}^{S_2}(r_1)\big| \big(T_j^{S_1} \in \Gamma_j^{S_1}\big) \cap \big(T_k^{S_1} \in \Gamma_k^{S_1}\big)\cap \big(R_1=r_1\big)\big\}\\
&=  \sum_{r_1=2}^p \rP(R_1=r_1)\sum_{j,k=1:j< k}^{p,p} \rP\big\{\big(T_j^{S_1} \in \Gamma_j^{S_1}\big) \cap \big(T_k^{S_1} \in \Gamma_k^{S_1}\big)\big|R_1=r_1\big\}\rP\big\{T_{jk}^{S_2} \in \Gamma_{jk}^{S_2}(r_1)\big\}\\
&= \sum_{r_1=2}^p \rP(R_1=r_1)\sum_{j,k=1:j< k}^{p,p} \rP\big\{\big(T_j^{S_1} \in \Gamma_j^{S_1}\big) \cap \big(T_k^{S_1} \in \Gamma_k^{S_1}\big)\big|R_1=r_1\big\}\frac{2\alpha}{r_1(r_1-1)}\\
&= \sum_{r_1=2}^p \rP(R_1=r_1) \sum_{j,k=1:j< k}^{p,p}\frac{r_1(r_1-1)}{p(p-1)}~~\frac{2\alpha}{r_1(r_1-1)}\\
&=\; \alpha\sum_{r_1=2}^p \rP(R_1=r_1) 
\;\leq\; \alpha.
\end{align*}  
The fifth (in)equality follows from the asymptotic pairwise  independence of the test-statistics in the first and second stage. 
\end{proof}


In the theorem it is assumed that, under the full null hypothesis, the vector $(T_j^{S_1},T_k^{S_1})$ and $T_{jk}^{S_2}$ are asymptotically independent. This assumption follows from the asymptotic pairwise independence proven in Result \ref{res: pairwise_teststatistics}, and the fact that the test-statistics $T_j^{S_1}$ and $T_k^{S_1}$ are asymptotically Gaussian under the null hypothesis.  

What is left, is to study the setting in which the test-statistics $T_j^{S_1}$ and $T_k^{S_1}$ for any $j,k, j\neq k$ are asymptotically independent. Define the models
\begin{align*}
\label{MM11}
M_1:\qquad h\big(\mathbb{E}(Y|\bX,\bbeta_1,\bnu_1)\big) &= \beta_{1,1} x_1 + \beta_{1,2} x_2 + \beta_{1,3} x_{j} \\
M_2:\qquad h\big(\mathbb{E}(Y|\bX,\bbeta_2,\bnu_2)\big) &= \beta_{2,1} x_1 + \beta_{2,2} x_2 + ~~~~~~~~~~~~\beta_{2,4} x_{k} 
\end{align*}
where $x_1=1$ for the intercept and $x_2$ the "fixed" covariate as before. 
The situation is exactly the same as for proving complete independence. It follows that for the Poisson, logistic and the linear regression model asymptotic independence holds if the covariates that are included in the two stage procedure (so not necessarily those that are permanently included in the model: $x_1$ and $x_2$) are mutually independent and independent of the ``fixed''  covariates, and the covariates are standardized to have zero mean, independence of the test statistics $T_j^{S_1}$ and $T_k^{S_1}$ based on the models $M_1$ and $M_2$ follows.

\section*{Appendix D: Pairwise independence in the Cox PH model}
In this Appendix the proof for pairwise independence of the regression parameters in the Cox model is given, The proof is in line with the proof of Theorem \ref{theorem:pairwuseGLM} in Subsection Appendix A, with the  partial log likelihood function in stead of the full log likelihood. However, since some nice properties of the family of GLMs with canonical link functions do not hold for the Cox model, the theorem only holds under the assumption of no main and interaction effects.\\ 

\noindent
{\bf Theorem \ref{Theorem:Cox pairwise} (Pairwise independence in the Cox PH model). } 
{\it
Let $(T_1=\tilde{T}_1\wedge C_1, \Delta_1, \bX_1), \ldots, (T_n=\tilde{T}_n\wedge C_1, \Delta_n, \bX_n)$ be a random sample. Suppose that $\tilde{T}_i|X_i$ is a random observation from the Cox PH model with hazard function as in (\ref{SM0}), then $\sqrt{n}(\hat\bbeta_1-\bbeta_1^0)$ and  $\sqrt{n}\hat\bbeta_2^\perp$ are asymptotically independent.
}

\begin{proof}
The partial log likelihood for the submodel $M_1$ in (\ref{SM1}) is given by
\begin{align*}
\frac{1}{n}PL_1(\bbeta_1) = \frac{1}{n} PL(E_1\beta_1) =  \frac{1}{n}\sum_{i=1}^n \int_0^\infty \log \Bigg(\frac{\exp(\bX_i^t E_1\bbeta_1)}{\sum_{j=1}^n Y_j(t)\exp(\bX_j^t E_1 \bbeta_1)}\Bigg){\rm d} N_i(t),
\end{align*}

For ease of notation, we define $S^{(0)}$ and its first and second derivative with respect to $\bbeta$, as
\begin{align*}
S^{(0)}(\bbeta,t) &:=\;\frac{1}{n}\sum_{i=1}^n Y_i(t) \exp(\bX_i^t \bbeta)     \\
S^{(1)}(\bbeta,t) &:=\;\frac{1}{n}\sum_{i=1}^n Y_i(t)  \bX_i \exp(\bX_i^t \beta)    \\ 
S^{(2)}(\bbeta,t) &:=\;\frac{1}{n}\sum_{i=1}^n Y_i(t)  \bX_i\exp(\bX_i^t \bbeta) \bX_i^t      
\end{align*}
Then, 
\begin{align*}
\frac{1}{n}PL_1(\bbeta_1) = \frac{1}{n} PL(E_1\bbeta_1) =  \frac{1}{n}\sum_{i=1}^n \int_0^\infty \log \Bigg(\frac{n^{-1}\exp(\bX_i^t E_1\bbeta_1)}{S^{(0)}(E_1 \bbeta_1,t)}\Bigg){\rm d} N_i(t).
\end{align*}

Define $\hat h_1=\hat\beta_1-\beta_1^0$. By a Taylor expansion of $PL_1$ with respect to $\hat h_1$ around zero
\begin{align*}
\frac{1}{n}\widetilde{PL}_1(\hat h_1) &:= \frac{1}{n}PL_1\big(E_1(\bbeta_1^0+\hat h_1)\big)\\[8pt]
&\propto \frac{1}{n} \sum_{i=1}^n \int_0^\infty \bX_i^t E_1 (\bbeta_1^0+\hat h_1) - \log\big(S^{(0)}(E_1(\bbeta_1^0+\hat h_1),t)\big) \; {\rm d} N_i(t)\\[8pt]
&= \frac{1}{n} \sum_{i=1}^n \int_0^\infty \bX_i^t E_1 (\bbeta_1^0+\hat h_1) - \log\big( S^{(0)}(E_1 \bbeta_1^0,t)\big) - \hat h_1^t E_1^t ~\frac{S^{(1)}(E_1\bbeta_1^0,t)}{S^{(0)}(E_1\bbeta_1^0,t)}\\[8pt]
& \qquad - \frac{1}{2}\hat h_1^t E_1^t \Big( \frac{S^{(2)}(E_1\beta_1^0,t)}{S^{(0)}(E_1\beta_1^0,t)} \;-\; \frac{S^{(1)}(E_1\beta_1^0,t)^tS^{(1)}(E_1\beta_1^0,t)}{S^{(0)}(E_1\beta_1^0,t)^2} \Big) E_1 \hat h_1  \; {\rm d} N_i(t) + O_p(\| \hat h_1\|^3)\\[8pt]
&= K_n + \hat h_1^t E_1^t A_n + \frac{1}{2}\hat h_1^t E_1^t B_n E_1\hat h_1 + O_p(\| \hat h_1\|^3)
\end{align*}
with
\begin{align*}
K_n &= \frac{1}{n} \sum_{i=1}^n \int_0^\infty \bX_i^t E_1 \bbeta_1^0 - \log\big(S^{(0)}(E_1 \beta_1^0,t)\big){\rm d}N_i(t) \\[8pt] 
A_n &= \frac{1}{n} \sum_{i=1}^n \int_0^\infty  \bX_i - \frac{S^{(1)}(E_1\bbeta_1^0,t)}{S^{(0)}(E_1\bbeta_1^0,t)} {\rm d}N_i(t) \\[8pt]
B_n &= -\frac{1}{n} \sum_{i=1}^n \int_0^\infty   \Big( \frac{S^{(2)}(E_1\bbeta_1^0,t)}{S^{(0)}(E_1\bbeta_1^0,t)} \;-\; \frac{S^{(1)}(E_1\bbeta_1^0,t)^t S^{(1)}(E_1\bbeta_1^0,t)}{S^{(0)}(E_1\bbeta_1^0,t)^2} \Big) {\rm d} N_i(t) 
\end{align*}

The term $K_n$ does not depend on $\hat h_1$. Furthermore,
$\hat h_1^t E_1^t A_n = n^{-1/2}\hat h_1^t E_1^t n^{1/2} A_n$
where
$n^{1/2} A_n$ equals $n^{-1/2}$ times the partial likelihood score function in $\beta_1^0$, and, by the martingale central limit theorem (see e.g., Andersen et al (1993)\nocite{Andersen}), is asymptotically normal with mean zero and variance matrix denoted as $I$. That means that
\begin{align*}
\hat h_1^t E_1^t A_n = \frac{1}{\sqrt{n}} \hat h_1^t E_1^t W_1 + O_p(1/\sqrt{n})
\end{align*}
with $W_1$ normally distributed with mean zero and covariance matrix $I$.
The last term $B_n$ converges in probability to the covariance matrix $I$ (see e.g., Andersen et al (1993)\nocite{Andersen}). That means that 
\begin{align*}
\hat h_1^t E_1^t B_n E_1\hat h_1 = \hat h_1^t E_1^t I E_1\hat h_1 + o_p(\|\hat h_1\|^2)  
\end{align*}
The same can be done for submodel $M_2$, yielding the variable $W_2$. Under the null hypothesis of no main and interaction effects, the covariance of $W_1$ and $W_2$ is equal to $\EE W_1W_2^t=\EE W_1 W_1^t=I$. 

The remainder of the proof follows the proof of Theorem \ref{theorem:pairwise} from Equation (\ref{limit in d average loglikelihood proof}) onwards.\\
\end{proof}

\section*{Appendix E: Power-plots for the Poisson regression model}
\label{sub:add_plots}

In this appendix we show the power-plots for the Poisson regression model for uncorrelated and correlated markers. The simulation study was performed as is described in subsection \ref{sub:simprocedure} and \ref{sub:power}.

\begin{figure}[h]
\begin{center}
\includegraphics[height=0.3\textwidth,width=0.3\textwidth]{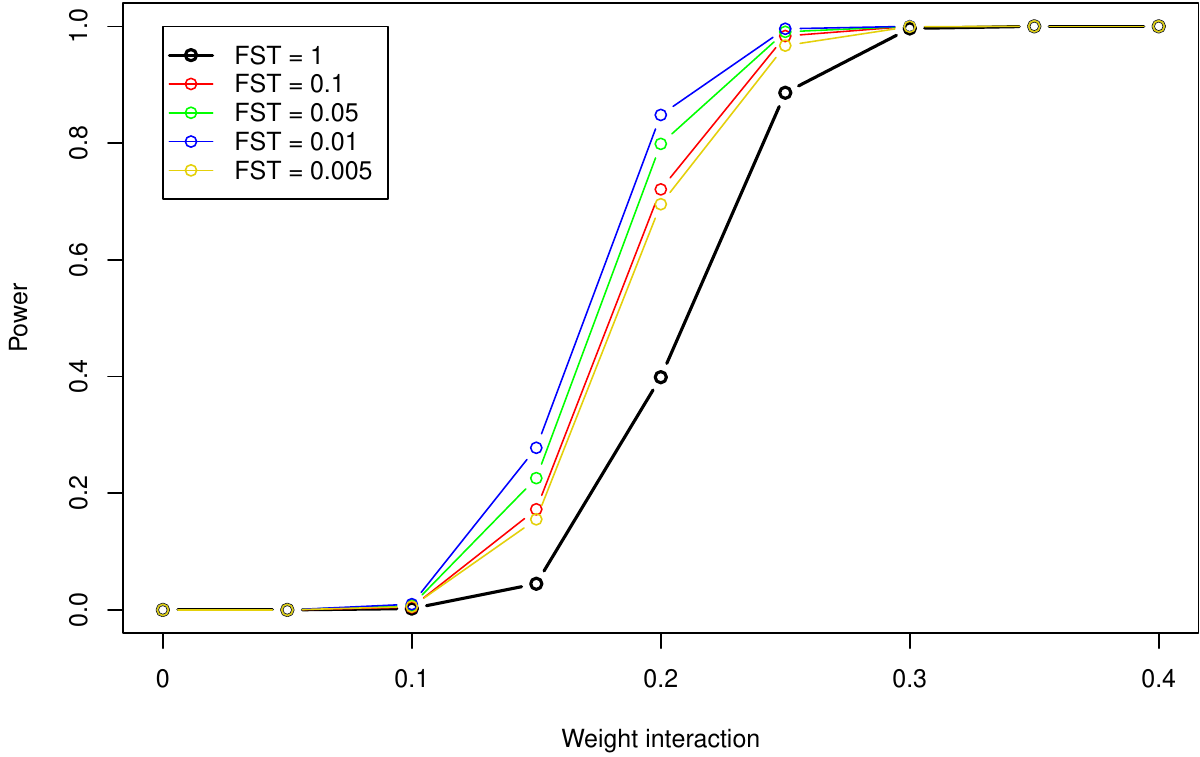}\;
\includegraphics[height=0.3\textwidth,width=0.3\textwidth]{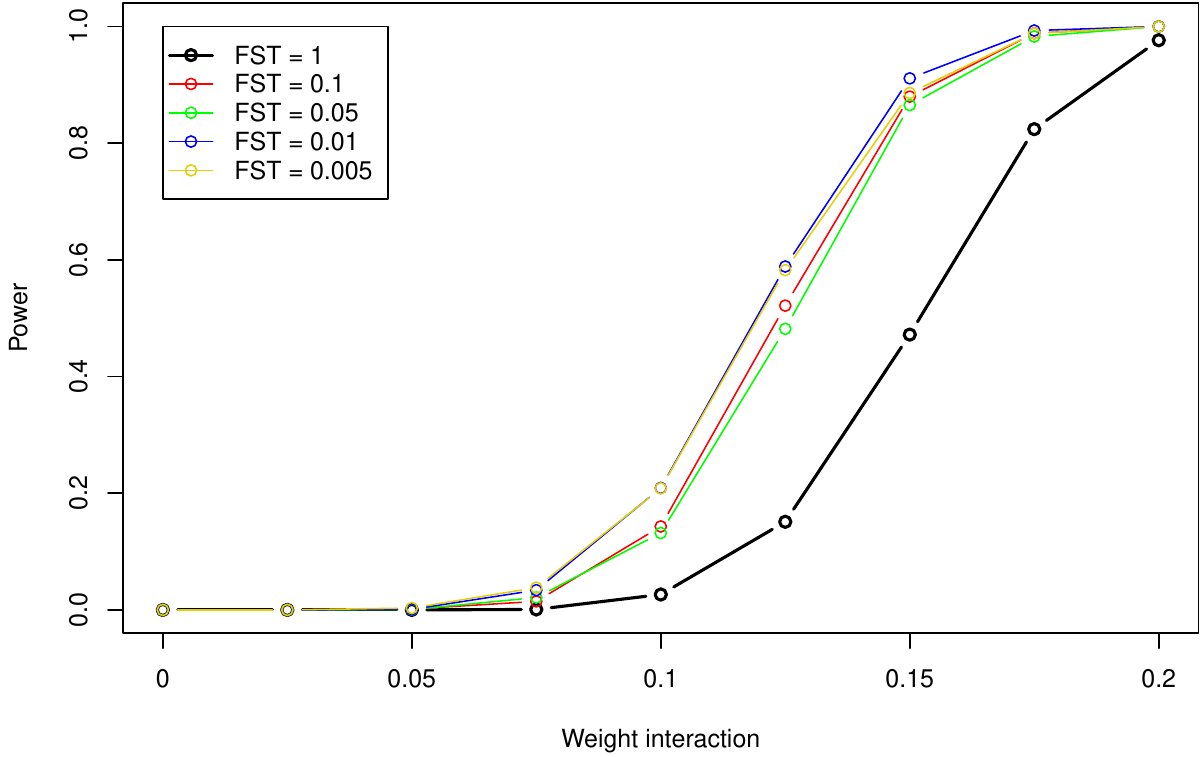} \;
\includegraphics[height=0.3\textwidth,width=0.3\textwidth]{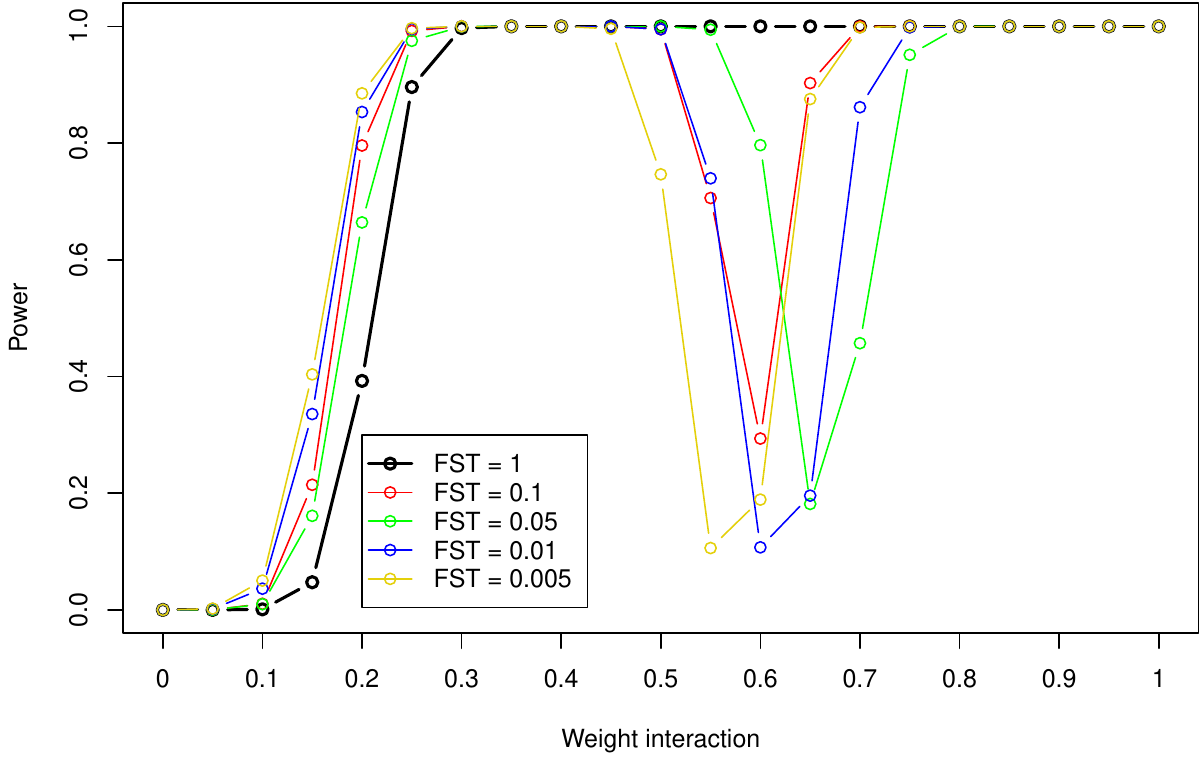}\\
\includegraphics[height=0.3\textwidth,width=0.3\textwidth]{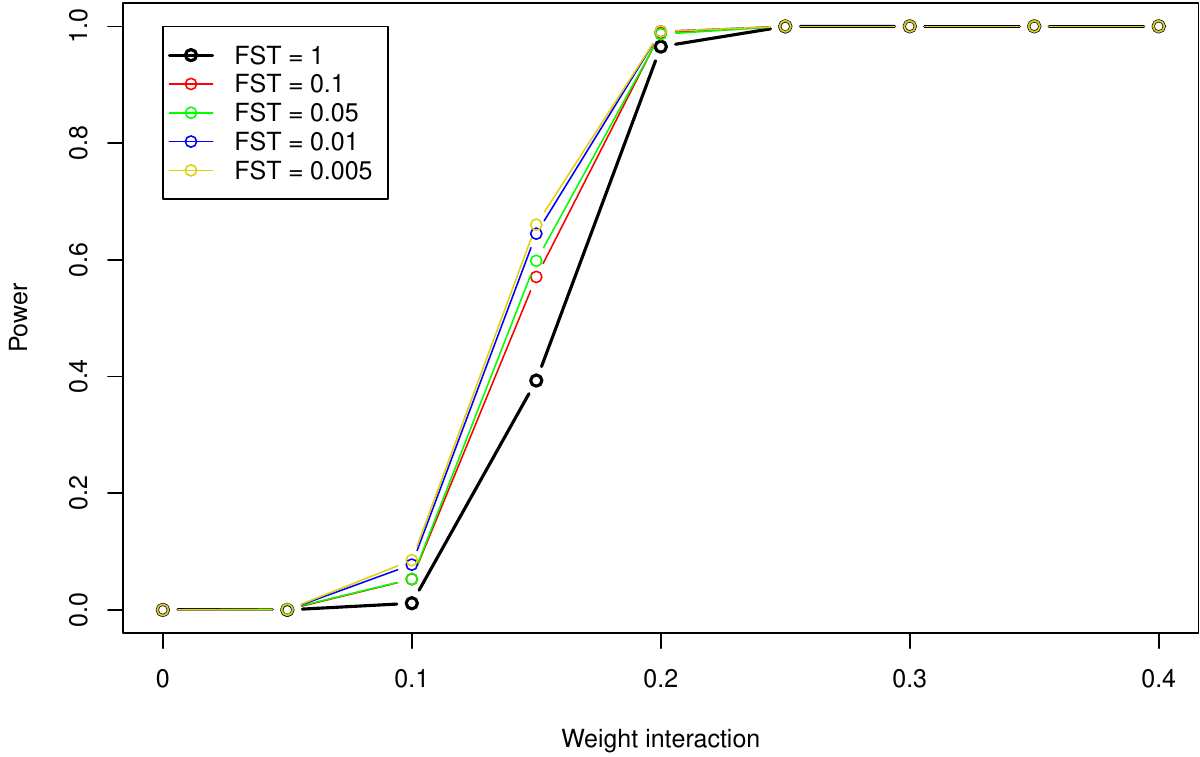}\;
\includegraphics[height=0.3\textwidth,width=0.3\textwidth]{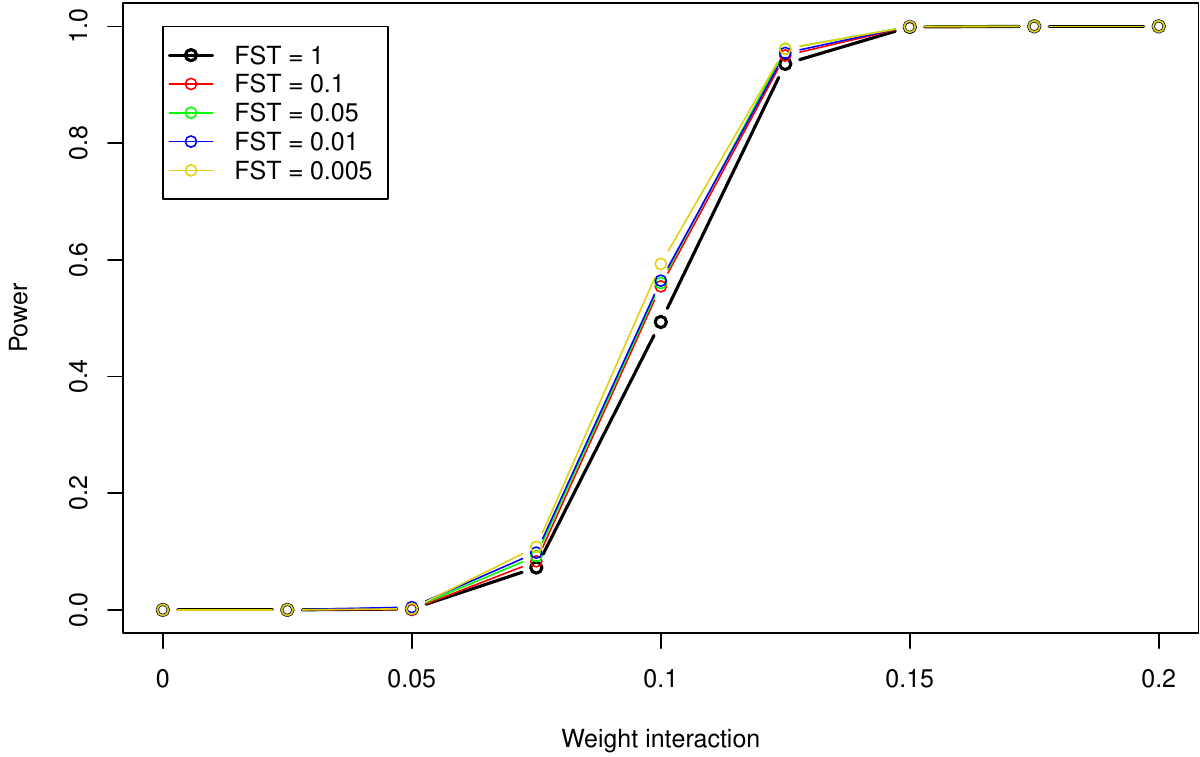} \;
\includegraphics[height=0.3\textwidth,width=0.3\textwidth]{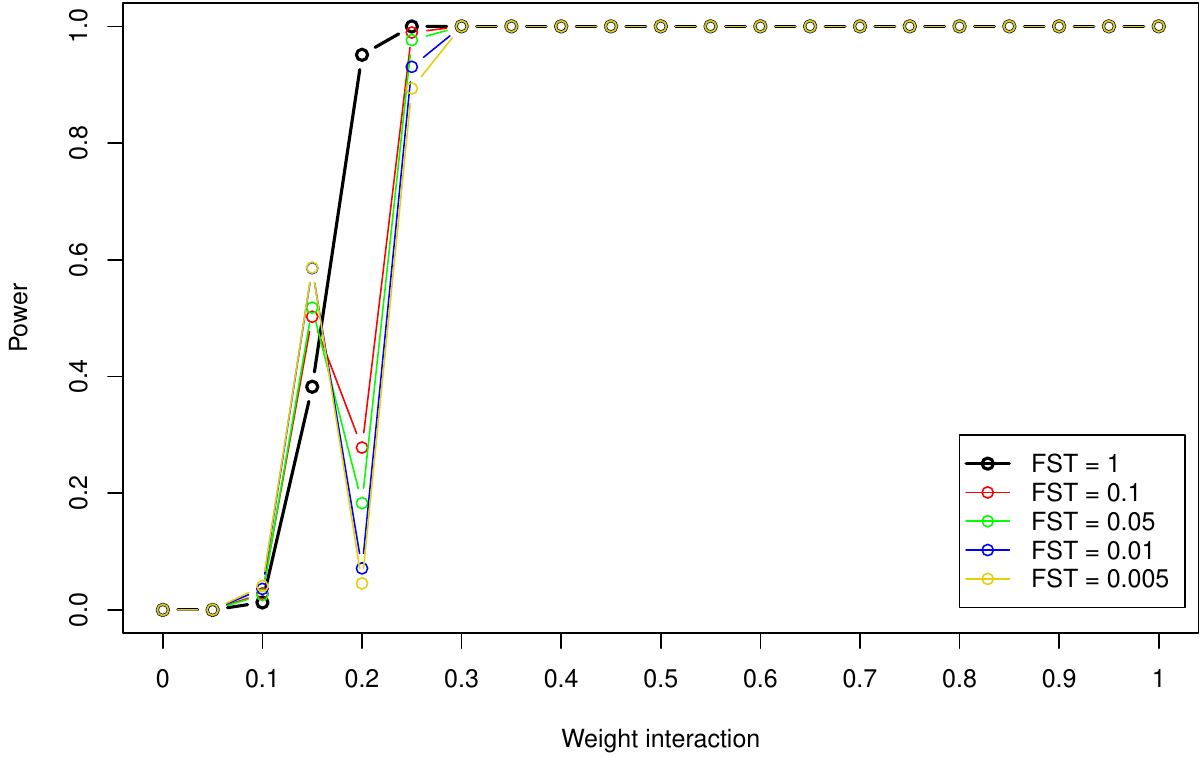} 
\caption*{Power as a function of the interaction effect for different first stage threshold (FST) for the Poisson regression model. First row: independent markers; Second row: correlated markers. First column: no main effects; Second column: both main effects equal to 0.5; Third column: both main effects equal to -0.5. (Note the different scales on the x-axes.)   }
\end{center}
\end{figure}
\end{document}